\documentclass[11pt]{article}
\usepackage{amssymb}
\usepackage{amsmath}
\usepackage{graphicx,color}
\usepackage{amsfonts}
\usepackage[ignoreall]{geometry}
\usepackage{natbib}

\newtheorem{theorem}{Theorem}

\newtheorem{corollary}[theorem]{Corollary}

\newtheorem{lemma}[theorem]{Lemma}

\newtheorem{method}[theorem]{Method}
\newtheorem{proposition}[theorem]{Proposition}

\newenvironment{proof}[1][Proof]{\noindent\textbf{#1.} }{\ \rule{0.5em}{0.5em}}
\title{Bootstrapping a Change-Point Cox Model for Survival Data}

\author{Gongjun Xu, Bodhisattva Sen and Zhiliang Ying\\ Columbia University}

\begin{document}
\maketitle

\begin{abstract}
This paper investigates the (in)-consistency of various bootstrap methods for making inference on a change-point in time in the Cox model with right censored survival data. A criterion is established for the consistency  of any bootstrap method. It is shown that the usual nonparametric bootstrap is inconsistent for the maximum partial likelihood estimation of the change-point.  A new model-based bootstrap approach is proposed and its consistency established. Simulation studies are carried out to assess the performance of various bootstrap schemes.
\end{abstract}
\section{Introduction}
The proportional hazards model of \cite{Cox72} specifies that the hazard function of survival time for a subject with possibly time-dependent covariate vector $Z$ is 
\begin{equation}\label{eq1}
\lambda(t | Z(s),s\leq t) = \exp\left( \beta'_0 Z(t)\right)
\lambda_0(t),
\end{equation}
where $\beta_0$ is a $p$-dimensional vector of regression parameters and $\lambda_0$ an unknown baseline hazard function. 
Inference on the regression parameter $\beta_0$ is usually based on the 
partial likelihood \citep{Cox1975}. The theoretical properties of the maximum partial likelihood estimator (MPLE) of $\beta_0$ have been studied extensively in the literature; see \cite*{AndersonGill}, \cite{FlemingHarrington}, and \cite{KalbfleischPrentice}.

It is sometimes plausible to postulate that the regression coefficient changes its value at a certain time, resulting in a change-point extension of the Cox model. For clinical trial data,  \cite{Meinert86} and \cite{Zucker90}  argue that the treatment effect may manifest only after a period of time. 
To model such lag effect, a two-phase Cox model with a change-point in time is usually considered
 and the hazard function is written as
\begin{equation}\label{cox}
\lambda(t | Z(s),s\leq t) = \exp\left(\alpha'_0 Z(t) 1_{t\leq \zeta_0} + \beta'_0 Z(t)1_{t > \zeta_0}\right) \lambda_0(t),
\end{equation}
where a second regression parameter vector $\alpha_0$ is added to model \eqref{eq1} and $\zeta_0$ is the change-point parameter. It is clear that estimation of the change-point $\zeta_0$ is an important step in the model based inference. For the identifiability of model \eqref{cox}, we assume throughout that $\alpha_0\neq \beta_0$ since otherwise  this model reduces to \eqref{eq1}  and $\zeta_0$ is not identifiable.

Model \eqref{cox} has been extensively studied in the literature. \cite*{Liang90} considered the problem of testing the null hypothesis of no change-point effect based on a maximal score statistic. 
 \cite*{Luo97} focused on testing $H_0: \zeta=\zeta_0$ v.s. $H_1: \zeta\neq\zeta_0$ for a pre-specified $\zeta_0$ and derived the asymptotic distribution of the  partial likelihood ratio test  statistic under $H_0$.
For estimation of the change-point parameter $\zeta_0$,  \cite{Luo1996} and \cite{Pons2002}  showed that the MPLE of  $\zeta_0$ is $n^{-1}$ consistent while that of the regression parameter vector is $n^{-1/2}$ consistent. This is largely due to the fact that the partial likelihood function is not differentiable with respect to the change-point parameter and therefore the usual Taylor expansion is not applicable. This ``nonstandard" asymptotic behavior of the MPLE of $\zeta_0$ is typical in change-point regression problems; see  \cite{kosorok2007}, \cite*{lan2009change}, and \cite{SeijoSen2011} for examples of different models.

Although the asymptotic distribution of the MPLE of $\zeta_0$ has been derived in the literature \citep{Luo1996,Pons2002}, it cannot be directly used for making inference for $\zeta_0$ due to the presence of  nuisance parameters.
Bootstrap methods bypass the difficulty of estimating the nuisance parameters
and are generally reliable in standard ${n^{-1/2}}$ convergence problems; see \cite{efron1993introduction} and \cite{boots1997}. When the bootstrap is applied to nonstandard problems such as the change-point model \eqref{cox}, however, it may yield invalid confidence intervals (CIs) for $\zeta_0$. 
The failure of the usual bootstrap methods in nonstandard situations has been documented  in the literature;  see  \cite{abrevaya2005bootstrap} and \cite*{sen2010inconsistency} for situations giving rise to $n^{1/3}$ asymptotic; see \cite*{bose2001generalised} for general $M$-estimation problems.
However, the change-point problem for the Cox model \eqref{cox} is indeed quite different from the problems considered by the above authors, and the performance of different bootstrap methods has not been investigated.  

Various bootstrap procedures have been applied to the standard Cox model \eqref{eq1} \citep{boots1997}. However, as to be shown in Section \ref{SectionCheck}, the commonly used bootstrap methods, such as sampling directly from the empirical  distribution  (ED) and sampling conditional on covariates \citep[]{Deborah1994}, provide invalid CIs for the change-point parameter in model \eqref{cox}.  Indeed, we show that the bootstrap estimates constructed by these methods are the smallest maximizers of certain stochastic processes and, conditional on the data, these processes do not have any weak limit. 
This strongly suggests not only the inconsistency but also the nonexistence of any weak limit for the corresponding bootstrap estimates. 

To get consistent bootstrap procedures, we develop a new bootstrap approach that, conditional on the covariates,  draws samples from a smooth approximation to the distribution of the survival time and from an estimate of the distribution of the censoring time. A key step in the new approach is the smooth approximation to the distribution of the survival time, which makes the bootstrap scheme successfully mimic the local behavior of the true distribution function at the location of $\zeta_0$. 
As a result, the proposed  approach yields asymptotically valid CIs for $\zeta_0$.
Furthermore, the asymptotic theory is also validated through simulation studies with reasonable sample sizes. 

The rest of this paper is organized as follows. In Section \ref{SectionModel} we describe the model setup and introduce different  bootstrap schemes. In Section \ref{SectionGeneral}, we state a series of convergence results. In Section \ref{SectionCheck} we study the inconsistency of the standard bootstrap methods, including sampling from the ED, and we prove the consistency of the smooth and the $m$-out-of-$n$ bootstrap procedures. We compare the finite sample performance of different bootstrap methods through a simulation study in Section \ref{SectionSim}.  Proofs of the main theorems are  in Section \ref{SectionProof}.  
Proofs of several lemmas are provided in the Appendix. 

\section{Model setup and bootstrap schemes}\label{SectionModel}
We use $T$ to denote survival time and $C$ censoring time. Throughout, $a\wedge b=\min\{a,b\}$ and $a\vee b =\max\{a,b\}.$
Let $\tilde T=T\wedge C$ and $\delta=1_{T\leq C}$ indicating failure
($1$) or censoring ($0$). 
Furthermore, there is a $p$-dimensional covariate process $Z(t)$, c\'agl\'ad (left-continuous with right-hand limits), which may include an individual's treatment assignment and certain relevant characteristics. 
In this paper we focus on the external time-dependent covariate and assume that $Z$ is observed over the study interval $[0,\tau]$, $\tau<\infty$. 
An external time-dependent covariate means that its value path is  not directly generated by the individual under the study. Examples include the age of an individual and the air pollution level for asthma study; see Chapter 6.3.1 in \cite{KalbfleischPrentice} for more discussion.
Furthermore, we assume that  $Z$ is of bounded total variation on $[0,\tau]$ and the covariance matrix $Var(Z(t))$ is strictly positive definite for any $t\in [0,\tau]$.

Given covariate $Z$, the survival time $T$ is assumed to be conditionally independent of the censoring time $C$. The hazard rate function of $T$ follows the change-point Cox model \eqref{cox} with $\alpha_0$ and $\beta_0$ belonging to bounded convex sets $\Theta_\alpha$ and $\Theta_\beta$ in $\mathbb{R}^{p}$, respectively. 
We assume that the baseline hazard function $\lambda_0(\cdot)$ is bounded on $[0,\tau]$ with $\inf_{t\in[0,\tau]}\lambda_0(t)>0$ and the conditional distribution of censoring time $G(\cdot|Z)$ satisfies $\sup_{z\in {\cal V}}G(\tau |z)<1$, where ${\cal V}$ is the set of all possible paths of $Z$.  To ensure the identifiability of $\zeta_0$, we further assume that $\lambda_0(\cdot)$ and  $G(\cdot|z)$, for any $z\in {\cal V}$, are continuous at $\zeta_0$. 

The observed data $(\tilde T_i, \delta_i, Z_i) ,i=1,...,n,$ consist of $n$ i.i.d. realizations of $(\tilde T, \delta, Z)$. The Cox partial likelihood \citep{Cox1975} is 
\begin{equation}\label{PL}
L_n(\alpha,\beta,\zeta) = \prod_{1\leq i \leq n, \delta_i=1} \frac{e^{\alpha' Z_i(T_i) 1_{T_i \leq \zeta} + \beta' Z_i(T_i)1_{T_i > \zeta}}}{\sum_{1\leq j\leq n, \tilde T_j\geq T_i}e^{\alpha' Z_j(T_i) 1_{T_i\leq \zeta} + \beta' Z_j(T_i)1_{T_i > \zeta}}}.
\end{equation}
Let $l_n(\alpha,\beta,\zeta)=\log L_n(\alpha,\beta,\zeta)$, which is continuous in $\alpha$ and $\beta$ but c\'{a}dl\'{a}g in $\zeta$. In fact, it is a step function in $\zeta$ and hence could have multiple maximizers. To avoid ambiguity, we say that $(\tilde \alpha_n',\tilde \beta_n',\tilde \zeta_n)' \in \Theta:=\Theta_\alpha\times\Theta_\beta\times[0,\tau]$ is a maximizer if
$$l_n(\tilde \alpha_n,\tilde \beta_n,\tilde \zeta_n-) \vee l_n(\tilde \alpha_n,\tilde \beta_n,\tilde \zeta_n) 
= \sup_{( \alpha', \beta', \zeta)'\in\Theta} l_n( \alpha, \beta, \zeta).$$
Since, for each $\zeta$, $l_n(\alpha, \beta,\zeta)$ as a function of  $\alpha$ and $\beta$ has a unique maximizer, we can choose as our MPLE the maximizer with the smallest value of $\zeta$. In other words, our estimator $(\hat \alpha_n',\hat \beta_n',\hat \zeta_n)'\in \Theta$ will be the only maximizer 
such that if $(\tilde \alpha_n',\tilde \beta_n',\tilde \zeta_n)'\in \Theta$ 
 is any other maximizer, then $\hat\zeta_n<\tilde \zeta_n.$
In this case, we say that $\hat\theta_n := (\hat \alpha_n',\hat \beta_n',\hat \zeta_n)'$ is the smallest argmax of $l_n$ and write it as
\begin{equation}\label{MPLEcox}
\hat\theta_n :=(\hat\alpha_n', \hat\beta_n',\hat\zeta_n)' 
:=  \hbox{sargmax}_{( \alpha', \beta', \zeta)'\in \Theta} l_n( \alpha, \beta, \zeta) .
\end{equation}

As discussed in the Introduction, it is not practical to directly use the limiting distribution of  
$\hat\zeta_n$ for constructing CIs for $\zeta_0$. Thus it is desirable to develop bootstrap approaches.

\subsection{Bootstrap procedures}\label{sectionbs}
We start with a brief review of  bootstrap procedures.
Consider a sample $\mathbf{X}_n = \{X_1,\cdots, X_n\}\stackrel{\text{iid}}{\sim}  F_X$. Suppose that we are interested in estimating the distribution function $F_{R_n}$ of a random variable $R_n(\mathbf{X}_n, F_X)$. 
A bootstrap procedure generates ${\mathbf X}_n^{*} = \{X_1^{*},\ldots, X_{m_n}^{*}\} \stackrel{\text{iid}}{\sim} \hat{F}_{X,n}$ given ${\mathbf X}_n$, where $\hat{F}_{X,n}$ is an estimator of $F_X$ from ${\mathbf X}_n$ and $m_n$ is a constant depending on $n$, and then estimates  $F_{R_n}$ by $F^*_{R_n}$, the conditional distribution function of $R_n({\mathbf X}_n^{*},\hat{F}_{X,n})$ given ${\mathbf X}_n$.
Let $d$ denote a metric metrizing  weak convergence of distributions. We say that $F^*_{R_n}$ is {\it weakly consistent} if $d(F_{R_n},F^*_{R_n})\stackrel{}{\rightarrow} 0$ in probability. If $F_{R_n}$ has a weak limit $F_{R}$, then weak consistency requires $F^*_{R_n}$ to converge weakly to $F_{R}$, in probability.

In the current context, we are interested in the distribution of $n(\hat\zeta_n-\zeta_0)$. Then for a consistent bootstrap procedure, the conditional distribution of $m_n(\hat\zeta^*_n-\hat\zeta_n)$ given the data must provide a good approximation to the distribution function of $n(\hat\zeta_n-\zeta_0)$,  where $\hat\zeta^*_n$ is the estimator of $\zeta_0$ obtained from the bootstrap sample. 
In the following we introduce several bootstrap methods commonly used in the literature for the model \eqref{cox}. We start with the classical bootstrap based on the  ED. 

\begin{method}[Classical bootstrap]\label{methodclassic}
Draw a random sample $\{(\tilde T^*_{n,i}, \delta^*_{n,i}, Z^*_{n,i}): i=1,\cdots,n\}$ from the ED of the data $\{(\tilde T_{i}, \delta_{i}, Z_{i}): i=1,\cdots,n\}$. 
\end{method}

An alternative to the usual nonparametric bootstrap method (Method \ref{methodclassic}) considered in non-regular problems is the $m$-out-of-$n$ bootstrap; see, e.g., \cite*{B1997}.

\begin{method}[$m$-out-of-$n$ bootstrap]\label{methodmofn}
 Choose an increasing sequence $\{m_n\}_{n=1}^\infty$ such that $m_n=o(n)$ and $m_n\rightarrow\infty$.  Draw a random sample $\{(\tilde T^*_{n,i}, \delta^*_{n,i}, Z^*_{n,i}): i=1,\cdots,m_n\}$ from the ED of the data $\{(\tilde T_{i}, \delta_{i}, Z_{i}): i=1,\cdots,n\}$. 
\end{method}

Two widely used conditional bootstrap procedures for the Cox model are given in Methods \ref{methodcond1} and \ref{methodcond2} below; see \cite{Deborah1994}. These methods are model-based and need estimators of the conditional distributions of $T$ and $C$ given $Z$.

\begin{method}[Bootstrap conditional on covariates]\label{methodcond1}~\newline
\indent 1. Fit the Cox regression model and construct an estimator of the conditional distribution of $T$ given $Z$ as
\begin{equation}\label{F}
\hat F^b_n( t | Z) =  1- \exp\Big(-\int_0^t e^{\hat\alpha_n' Z(s) 1_{s\leq \hat\zeta_n} +\hat \beta'_n Z(s)1_{s > \hat\zeta_n}}d\hat\Lambda^b_{n,0}(s)\Big),
\end{equation}
where $\hat\Lambda^b_{n,0}(s)$  is the Breslow estimator of the cumulative baseline hazard function $\Lambda_0$, i.e.,
$$\hat\Lambda^ b_{n,0}(t)  = \int_0^t \Big(\sum_{j=1}^n Y_j(s)e^{\hat\alpha_n' Z_j(s) 1_{s\leq \hat\zeta_n} +\hat \beta'_n Z_j(s)1_{s > \hat\zeta_n}}\Big)^{-1}d\Big(\sum_{i=1}^n N_i(s)\Big)$$
with $Y_{i}(t) = 1_{\tilde{T}_{i}\geq t}$ and $N_{i}(t) =  1_{\tilde T_{i}\leq t, \delta_i=1}.$ 
In addition, we construct a conditional distribution estimator $\hat G_n(\cdot |Z)$ of $G(\cdot | Z)$; see Section \ref{sectionbs2} for more discussion on estimating  $G(\cdot|Z)$. \\
\indent 2. For given $Z_1, \cdots, Z_n$, generate i.i.d. replicates $\{T^*_{n,i}, C^*_{n,i}: i=1,\cdots,n\}$ from the conditional distribution estimators $\{\hat F^b_n(\cdot |Z_i), \hat G_n(\cdot |Z_i): i=1,\cdots, n\}$, respectively. 
Then we obtain a bootstrap sample
$\{(\tilde T^*_{n,i}, \delta^*_{n,i}, Z_{i}): i=1,\cdots,n\}$, where $\tilde T^*_{n,i} = T^*_{n,i}\wedge C^*_{n,i}$ and  $\delta^*_{n,i}= 1_{T^*_{n,i}\leq C^*_{n,i}}$.
\end{method}

\begin{method}[Bootstrap conditional on  covariates and censoring]\label{methodcond2}~\newline
\indent 1. Same as Step 1 in Method \ref{methodcond1}.\\
\indent 2. For given $Z_1, \cdots, Z_n$, generate $T^*_{n,i}$ from $\hat F^b_n(\cdot |Z_i)$. If $\delta_i =0$, let $C^*_{n,i} =  C_i$;  otherwise, generate $C^*_{n,i}$  from $\hat G_n(\cdot|Z_i)$ conditioning on $C^*_{n,i}>T_i$. 
\end{method}

Methods \ref{methodclassic}, \ref{methodcond1} and \ref{methodcond2} are the most widely used bootstrap methods for the Cox regression model. In the following sections, we demonstrate through theoretical derivation and simulation the inconsistency of these methods for constructing CIs for $\zeta_0$. 
To get a consistent estimate of the distribution of $n(\hat\zeta_n-\zeta_0)$,
we propose the following smooth bootstrap procedures. 

\begin{method}[Smooth bootstrap conditional on covariates]\label{methodsmooth1}~\newline
\indent 1.Choose an appropriate nonparametric smoothing procedure (e.g., kernel estimation method in \cite{Wells94}) to build an estimator $\hat\lambda_{n,0}$ of $\lambda_0$. 
The associated estimator of $F(t|Z)$ is
\begin{equation}\label{FS}
\hat F^s_n( t | Z) =  1- \exp\Big(-\int_0^t e^{\hat\alpha_n' Z(s) 1_{s\leq \hat\zeta_n} +\hat \beta'_n Z(s)1_{s > \hat\zeta_n}}\hat\lambda_{n,0}(s)ds\Big).
\end{equation}
\indent 2. For given $Z_1, \cdots, Z_n$, generate i.i.d. replicates $\{T^*_{n,i}, C^*_{n,i}: i=1,\cdots,n\}$ from the conditional distribution estimatiors $\{\hat F^s_n(\cdot |Z_i), \hat G_n(\cdot |Z_i):  i=1,\cdots, n\}$, respectively. Then we obtain a bootstrap sample $\{\tilde T^*_{n,i}, \delta^*_{n,i}, Z_i: i=1,\cdots,n\}$. 
\end{method}

\begin{method}[Smooth bootstrap conditional on  covariates and  censoring]\label{methodsmooth2}~\newline
\indent 1. Same as Step 1 in Method \ref{methodsmooth1}.\\
\indent 2. For given $Z_1, \cdots, Z_n$, generate $T^*_{n,i}$ from $\hat F^s_n(\cdot |Z_i)$. If $\delta_i =0$, let $C^*_{n,i} =  C_i$;  otherwise, generate $C^*_{n,i}$  from $\hat G_n(\cdot|Z_i)$ conditioning on $C^*_{n,i}>T_i$. 
\end{method}

We will use a general convergence result established in Section \ref{SectionGeneral} to prove that the smooth bootstrap procedures (Methods \ref{methodsmooth1} and \ref{methodsmooth2}) and the $m$-out-of-$n$ procedure (Method \ref{methodmofn}) are  consistent. 
We will also illustrate through a simulation study that  the smooth bootstrap methods  outperform the  $m$-out-of-$n$ method. 

\section{A general  convergence result}\label{SectionGeneral}
In this section we prove a general convergence theorem for triangular arrays of random variables in the non-regular Cox proportional hazard model with a change-point in time. This theorem will be applied to show the consistency of the bootstrap procedures introduced in the previous section. 

We first introduce some notation.
Let $\mathbb{P}$ be a distribution satisfying the change-point Cox model \eqref{cox} for some parameter 
$\theta_0:=(\alpha_0',\beta_0',\zeta_0)'\in \Theta:=\Theta_\alpha\times\Theta_\beta\times[0,\tau]$. 
 Consider a triangular array of independent random samples 
$\{(\tilde T_{n,i},\delta_{n,i}, Z_{n,i}): i=1,\cdots, m_n\}$ defined on a probability space $(\Omega, \cal A, {\bf P})$, where $\tilde T_{n,i}=T_{n,i}\wedge  C_{n,i}$, $\delta_{n,i}=1_{T_{n,i}\leq  C_{n,i}}$, and $m_n\rightarrow \infty$ as $n\rightarrow\infty$. We use $\mathbf{E}$ to denote the expectation operator with respect to ${\bf P}$. Furthermore, we assume that $\{(\tilde T_{n,i}, \delta_{n,i}, Z_{n,i}): i=1,\cdots, m_n\}$ jointly follows a distribution $\mathbb{Q}_{n}$, and for each $i$,  the distribution of $(\tilde T_{n,i}, \delta_{n,i}, Z_{n,i})$ is $\mathbb{Q}_{n,i}$. 

As in Section \ref{SectionModel}, we assume that under $\mathbb{Q}_n$, the covariate process $Z(t)$ is c\'agl\'ad and has bounded total variation on $[0,\tau]$. 
We write $Z^{\otimes 0}=1$, $Z^{\otimes
1}=Z$, and $Z^{\otimes 2}=ZZ'$.
For the $i$th subject, let
$Y_{n,i}(t) = 1_{\tilde{T}_{n,i}\geq t}$ and $N_{n,i}(t) =  1_{\tilde T_{n,i}\leq t, \delta_{n,i}=1}.$
For $\gamma\in \mathbb{R}^p$ and $k=0,1$ and $2$,  let
\begin{eqnarray*}\label{gamma}
S_{n,k}(t;\gamma) &=& \frac{1}{m_n}\sum_{i=1}^{m_n}Y_{n,i}(t)Z_{n,i}^{\otimes k}(t)\exp(\gamma' Z_{n,i}(t)),\\
s_{n,k}(t;\gamma) &=&  \mathbb{Q}_n \Big(\frac{1}{m_n}\sum_{i=1}^{m_n}Y_{n,i}(t)Z_{n,i}^{\otimes k}(t)\exp(\gamma' Z_{n,i}(t))\Big),\\
s_{k}(t;\gamma) &=& \mathbb{P} \left(Y(t)Z^{\otimes k}(t)\exp(\gamma' Z)\right),\\
A_{n,k}(t) &= &
\mathbb{Q}_n\Big(\frac{1}{m_n}\sum_{i=1}^{m_n} \int_0^t Z_{n,i}^{\otimes k}(s) dN_{n,i}(s)\Big),\\
A_k(t) &=& \mathbb{P}\Big(\int_0^t Z^{\otimes k}(s) dN(s)\Big) 
= \int_0^t s_k(s; \alpha_0 1_{s\leq \zeta_0} + \beta_0 1_{s> \zeta_0} ) \lambda_0(s)ds,
 \end{eqnarray*}
where we use $\mathbb{Q}_n(\cdot)$ and $\mathbb{P}(\cdot)$ to denote the expectation operators under the distributions $\mathbb{Q}_n$ and $\mathbb{P}$, respectively. 
We write 
$$\bar Z_{n}(t;\gamma) = \frac{S_{n,1}(t;\gamma)}{S_{n,0}(t;\gamma)},~~
\bar z_{n}(t;\gamma) = \frac{s_{n,1}(t;\gamma)}{s_{n,0}(t;\gamma)},~~
\bar z(t;\gamma) = \frac{s_{1}(t;\gamma)}{s_{0}(t;\gamma)}.$$
Further we denote the ratio between ${S_{n,0}(t;\gamma_1)} $ and ${S_{n,0}(t;\gamma_2)} $ by 
\begin{align*}
&R_{n}(t;\gamma_1,\gamma_2) = \frac{S_{n,0}(t;\gamma_1)}{S_{n,0}(t;\gamma_2)} . 
\end{align*}
Similarly we write 
  \begin{align*}
r_{n}(t;\gamma_1,\gamma_2) = \frac{s_{n,0}(t;\gamma_1)}{s_{n,0}(t;\gamma_2)},~~
 ~ r(t;\gamma_1,\gamma_2) = \frac{s_{0}(t;\gamma_1)}{s_{0}(t;\gamma_2)}. 
  \end{align*}


Using the above notation, for $\theta=(\alpha',\beta',\zeta)'$,  the log partial likelihood function of $\{(\tilde T_{n,i}, \delta_{n,i}, Z_{n,i}): i=1,\cdots, m_n\}$ takes the form
 \begin{eqnarray*}
l^*_n(\theta) =  \sum_{i=1}^{m_n}
 \int_0^\tau \Big((\alpha 1_{s\leq \zeta} + \beta 1_{s> \zeta} )'Z_{n,i}
 - \log S_{n,0} (s; \alpha 1_{s\leq \zeta} + \beta 1_{s> \zeta})\Big)
dN_{n,i}(s).
\end{eqnarray*}
Denote the MPLE of $l^*_n(\theta)$ by $\theta^*_n=({\alpha^*_n}', {\beta^*_n}',\zeta^*_n)' $, i.e.,
$$ \theta^*_n :=  \hbox{sargmax}_{\theta\in \Theta}
l^*_n( \theta) .$$
Let $\theta_n=(\alpha_n', \beta_n',\zeta_n)' $ be given by 
\begin{eqnarray*}
&& \theta_n 
:= \hbox{sargmax}_{\theta \in \Theta}
 \mathbb{Q}_{n}\biggr(\frac{1}{m_n}\sum_{i=1}^{m_n}
 \int_0^\tau \big((\alpha 1_{s\leq \zeta} + \beta 1_{s> \zeta} )'Z_{n,i}\\
&&~~~~~~~~~~~~~~~~~~~~~~~~ - \log s_{n,0} (s; \alpha 1_{s\leq \zeta} + \beta 1_{s> \zeta})\big)
dN_{n,i}(s) \biggr).
\end{eqnarray*}
The existence of $\theta_n$ is guaranteed as the above objective function is concave in $\alpha$ and $\beta$ for every fixed $\zeta$ and bounded and c\'adl\'ag as a function of $\zeta$. 
When $\mathbb{Q}_n$ is the ED of a sample generated from model \eqref{cox}, $\theta_n$ becomes the usual MPLE $\hat\theta_n$ of $l_n(\theta)$ as defined in \eqref{MPLEcox}.

In the following, we derive sufficient conditions on the distribution ${\mathbb Q}_n$ that  guarantees the weak convergence of $(\sqrt{m_n}(\alpha_n^*-\alpha_n)', \sqrt{m_n}(\beta_n^*-\beta_n)', {m_n}(\zeta^*_n-\zeta_n))'$.

\subsection{Consistency and the rate of convergence}

We first show the consistency of the MPLE $\theta_n^*$ of $l_n^*$, whose proof is given in Section \ref{SectionProof}. We need the following assumption.

\begin{itemize}
\item[A1.]   For $k=0,1$ and $2$,  as $n\rightarrow\infty$,
$$\sup_{t\in [0,\tau], \gamma\in \Theta_\alpha\cup\Theta_\beta}|s_{n,k}(t;\gamma) - s_{k}(t;\gamma)| 
 \xrightarrow[]{} 0~
\mbox{ and } ~
\sup_{t\in [0,\tau]}\left|A_{n,k}(t) - A_k(t) \right| \xrightarrow[]{} 0,$$
where $|\cdot|$ denotes the $L_1$ norm. 
\end{itemize}

Condition A1 indicates that $\mathbb{Q}_n$ approaches  the distribution satisfying the Cox model \eqref{cox} in the sense that  the difference between expectations of $S_{n,k}$ ($A_{n,k}$) under distributions $\mathbb{Q}_n$ and $\mathbb{P}$ goes to 0 as $n\rightarrow \infty$. 
When $\mathbb{Q}_n$ is the ED of a sample from model \eqref{cox}, the uniform law of large numbers implies A1; see Section \ref{sectionbs1} for more details.

\begin{theorem}\label{theorem1}
Under condition $A1$, $\theta_n  \xrightarrow[]{}\theta_0$ and $\theta_n^* \xrightarrow[]{\mathbf{P}} \theta_0$.
\end{theorem}

We consider the rate of convergence of $\theta^*_n$ and show that the estimators of the ``regular'' parameters, $\alpha^*_n$ and $\beta^*_n$, converge at a rate of ${m_n}^{-1/2}$ while the change-point $\zeta^*_n$ converges at rate  $m_n^{-1}$. To guarantee the right rate of convergence, we need the following condition. 

\medskip
A2. There exist positive constants $\rho_1$ and $\rho_2$ such that, for any sequence $\{h_n\}$ satisfying $ h_n\rightarrow\infty$ and $h_n/m_n \rightarrow 0$ as $n\rightarrow\infty$, 
the following holds:
\begin{eqnarray*}
 \lim_{n\rightarrow\infty}\sup_{|h|> \frac{h_n}{m_n}}
\frac{1}{h}\biggr|\int_{\zeta_n}^{\zeta_n+h} dA_{n,1}(s)
- \int_{\zeta_n}^{\zeta_n+h}{\bar z_{n}(s; \alpha_n 1_{s\leq \zeta_n} + \beta_n 1_{s> \zeta_n})}
 dA_{n,0}(s)\biggr|
= 0,
\end{eqnarray*} 
and
\begin{eqnarray*}
&& \rho_1<\lim_{n\rightarrow\infty}\inf_{|h|> \frac{h_n}{m_n}}\frac{1}{h}
\left |{A_{n,0}(\zeta_n+h) - A_{n,0}(\zeta_n) }\right|\\
&&~~~~~~~~~~~~~~~~~~~~\leq
\lim_{n\rightarrow\infty}\sup_{|h|> \frac{h_n}{m_n}}\frac{1}{h}
\left |{A_{n,0}(\zeta_n+h) - A_{n,0}(\zeta_n) }\right|
<\rho_2. 
\end{eqnarray*}

 Note that condition A2 holds if  under $\mathbb{Q}_n$ the survival time  $T$ has uniformly bounded baseline hazard rate function $\lambda_{n,0}$ in some neighborhood of $\zeta_n$. In this case, $A_{n,0}$ has right derivative $s_{n,0}(\zeta_n;\beta_n)\lambda_{n,0}(\zeta_n)$ at $\zeta_n^+$ and left derivative $s_{n,0}(\zeta_n;\alpha_n)\lambda_{n,0}(\zeta_n)$ at $\zeta_n^{-}$, which implies A2. 

\begin{theorem}\label{theorem2}
Under conditions $A1$ and $A2$, 
$$\big|\big(\sqrt{m_n}(\alpha_n^*-\alpha_n)', \sqrt{m_n}(\beta_n^*-\beta_n)', {m_n}(\zeta^*_n-\zeta_n)\big)' \big|
= O_{\mathbf{P}}(1).$$
\end{theorem}

\subsection{Asymptotic distribution}
To compute the asymptotic distribution of $\theta^*_n$, we need the following assumption. 

\begin{itemize}
\item[A3.] For any $t\in \mathbb{R}$, $h_1<h_2$ and $0\notin (h_1,h_2)$,
\begin{eqnarray*} 
{\mathbb Q}_n\biggr(\sum_{k=1}^{m_n}
\int_{\zeta_n+h_1/m_n}^{\zeta_n+h_2/m_n} e^{\imath t (\alpha_n-\beta_n)'Z_{n,k}(s)} dN_{n,k}(s)\biggr)
~~~~~~~~~~~~\\
~~~~~~~~~~~~\rightarrow 
s_{0}\left(\zeta_0;\gamma_0+\imath	t(\alpha_0-\beta_0)\right)\lambda_0(\zeta_0)(h_2-h_1),
\end{eqnarray*}
where $\imath$ is the imaginary unit and $\gamma_0 = \alpha_01_{h_2\leq 0}+\beta_01_{0\leq h_1}.$
\end{itemize}

Condition A3 holds if  under $\mathbb{Q}_n$ the survival time  $T$ has uniformly bounded baseline hazard rate function $\lambda_{n,0}$  converging uniformly to $\lambda_0$ in some neighborhood of $\zeta_0$. This is satisfied by the smooth bootstrap methods introduced in Section \ref{sectionbs} and therefore guarantees their consistency; see Section \ref{sectionsmooth} for more details. 

\smallskip
We write $X_n(\theta) = {m_n}^{-1}( l^*_n(\theta) - l^*_n(\theta_n)).$
Note that  $\theta^*_n$ is also the maximizer of $X_n$. For $h:=(h_\alpha',h_\beta',h_\zeta)' \in\mathbb{R}^{p}\times \mathbb{R}^{p}\times \mathbb{R}$, consider the multiparameter process 
$$U^*_n(h):= m_n X_n\Big(\alpha_n+\frac{h_\alpha}{\sqrt{m_n}},\beta_n+\frac{h_\beta}{\sqrt{m_n}},\zeta_n+\frac{h_\zeta}{m_n}\Big) $$
and observe that 
\begin{equation}\label{argmaxu}
(\sqrt{m_n}(\alpha^*_n-\alpha_n)', \sqrt{m_n}(\beta^*_n-\beta_n)', {m_n}(\zeta^*_n-\zeta_n))' = \mbox{sargmax}_{h \in\mathbb{R}^{2p+1}}  U^*_n(h).
\end{equation}

We proceed to describe the limit law of the process $U^*_n$. Let 
$\Gamma^-$ and $\Gamma^+$ be two homogenous Poisson processes with intensities 
$$\gamma^- = s_0(\zeta_0; \alpha_0)\lambda_0(\zeta_0) \mbox{ and } \gamma^+ = s_0(\zeta_0; \beta_0)\lambda_0(\zeta_0),$$ respectively. 
Define two  sequences of i.i.d. random variables $\mathbf{v}^{-} = (v_i^-)_{i=1}^\infty$ and $\mathbf{v}^{+} = (v_i^+)_{i=1}^\infty$ such that $v_i^-$ and $v_i^+$ follow distributions:
 $$\mathbf{P}(v_i^- \leq z) = s_{0}(\zeta_0;\alpha_0)^{-1} 
 \mathbf{E}\left[1_{Z(\zeta_0)\leq z} Y(\zeta_0)e^{\alpha_0'Z(\zeta_0)}\right]$$
 and 
 $$\mathbf{P}(v_i^+ \leq z) = s_{0}(\zeta_0;\beta_0)^{-1} 
 \mathbf{E}\left[1_{Z(\zeta_0)\leq z} Y(\zeta_0)e^{\beta_0'Z(\zeta_0)}\right].$$
Additionally, take two Gaussian $\mathbb{R}^p$-valued random vectors 
$$U_1\sim N\biggr(0,\int_0^{\zeta_0} Q(s;\alpha_0)s_0(s;\alpha_0)\lambda_0(s)ds\biggr), 
~U_6\sim N\biggr(0,\int_{\zeta_0}^\tau Q(s;\beta_0)s_0(s;\beta_0)\lambda_0(s)ds\biggr),$$
where for $s\in\mathbb{R}$  and $\gamma\in \mathbb{R}^p$, $Q(s;\gamma)$ is defined as
 \begin{eqnarray}\label{QQQ}
 Q(s;\gamma) = \frac{s_{2}(s; \gamma)} {s_{0}(s; \gamma)} 
- {\bar z(s; \gamma)}^{\otimes 2}.
\end{eqnarray}
 
Suppose that $\Gamma^-$,$\Gamma^+$, $\mathbf{v}^{-}, \mathbf{v}^{+}$, $U_1$ and $U_6$ are all independent.
For $h_\zeta\in \mathbb{R}$, define the vector-valued process $(U_2,U_3,U_4, U_5)$ as
\begin{eqnarray*}
U_2(h_\zeta) &:=& 1_{h_\zeta<0} \Big(\Gamma^-(-h_\zeta)
\log{r(\zeta_0;\alpha_0,\beta_0)} 
+\sum_{1\leq i\leq \Gamma^-(-h_\zeta)}(\beta_0-\alpha_0) v_i^- \Big),\\
U_3(h_\zeta) &:=& 1_{h_\zeta<0} \Gamma^-(-h_\zeta),\\
U_4(h_\zeta) &:=& 1_{h_\zeta>0} \Big(\Gamma^+(h_\zeta)
\log{r(\zeta_0;\beta_0,\alpha_0)} +\sum_{1\leq i\leq \Gamma^+(h_\zeta)}(\alpha_0-\beta_0) v_i^+ \Big),\\
U_5(h_\zeta) &:=& 1_{h_\zeta>0} \Gamma^+(h_\zeta).
\end{eqnarray*}
Furthermore, define processes $J(h_\zeta) := U_{3}(h_\zeta) + U_{5}(h_\zeta)$ and
\begin{eqnarray*}
 U(h_\alpha, h_\beta,h_\zeta) &:=& h_\alpha' U_{1} - \frac{1}{2}h_\alpha'
 \biggr(\int_0^{\zeta_0} Q(s; \alpha_0)s_0(s;\alpha_0)\lambda_0(s)ds\biggr)h_\alpha + U_{2}(h_\zeta)\\
&&+ h_\beta' U_{6} - \frac{1}{2}h_\beta'
\biggr(\int_{\zeta_0}^\tau Q(s; \beta_0)s_0(s;\beta_0)\lambda_0(s)ds\biggr)h_\beta + U_{4}(h_\zeta).
\end{eqnarray*}
Observe that $J$ is the sequence of jumps of $U$. 
Our goal is to show that the asymptotic distribution of the MPLE is exactly that of the smallest argmax of $U$. Before doing this, we state the following result about the smallest argmax of $U$. 

\begin{lemma}\label{lemma4}
Let  $\phi = (\phi_\alpha',\phi_\beta',\phi_\zeta)' = sargmax_{h\in\mathbb{R}^{2p+1}}U(h)$ with $\phi_\alpha,\phi_\beta$ and $\phi_\zeta$ corresponding to the first $p$, the second $p$ and the last component of $\phi$, respectively. Then $\phi$ is well-defined. Moreover, $\phi_\alpha,\phi_\beta$ and $\phi_\zeta$ are mutually independent and 
\begin{eqnarray}\label{dis1}
\phi_\alpha &\sim& N\biggr(0,\Big(\int_0^{\zeta_0} Q(s;\alpha_0)s_0(s;\alpha_0)\lambda_0(s)ds\Big)^{-1}\biggr), \\\label{dis2}
\phi_\beta &\sim& N\biggr(0,\Big(\int_{\zeta_0}^\tau Q(s;\beta_0)s_0(s;\beta_0)\lambda_0(s)ds\Big)^{-1}\biggr).
\end{eqnarray}
\end{lemma}

We are now in a position to give our main result.
To state the result, we need to introduce some further notation. For any given compact set $K\subset \mathbb R^{d}, d\in \mathbb{N}$, we define the space ${\cal D}_K$ as the Skorohod space of functions $f: K\rightarrow \mathbb{R}$ having ``quadrant limits'' and continuous from above; see \cite{Neuhaus} and  \cite{SeijoSen2011} for more information about this space. Further, we take ${\cal D}_K$ as a metric space endowed with the Skorohod metric, which ensures the existence of conditional probability distributions for its random elements; see \cite{Neuhaus} and  Theorem 10.2.2 of \cite{MR1932358}.  
\begin{theorem}\label{theorem3}
Under conditions $A1$-$A3$, 
for a compact rectangle $\Theta \subset \mathbb{R}^{2p+1}$,
$U^*_n$ converges weakly in the Skorohod topology to $U$ in ${\cal D}_\Theta$. Moreover,
$$\left(\begin{array}{c} \sqrt{m_n}(\alpha_n^*-\alpha_n) \\ \sqrt{m_n}(\beta_n^*-\beta_n)\\
m_n(\zeta_n^*-\zeta_n)\end{array}\right) \leadsto sargmax_{h\in\mathbb{R}^{2p+1}} {U(h)},
$$
where $\leadsto$ denotes weak convergence. 
\end{theorem}

We  consider the MPLE $\hat\theta_n$ of $l_n(\theta)$ as defined in \eqref{MPLEcox}.
In this case, we can take $m_n=n$, $\mathbb{Q}_{n,i} = \mathbb{P}$ and $\theta_n = \theta_0$. Then conditions A1-A3 automatically hold and we immediately obtain the following corollary from Theorem \ref{theorem3}; see also \cite{Pons2002}. 
\begin{corollary}
Under the model setup in Section \ref{SectionModel}, for the MPLE $\hat\theta_n=(\hat\alpha_n',\hat\beta_n',\hat\zeta_n)'$, 
$$\left(\begin{array}{c} \sqrt{n}(\hat\alpha_n-\alpha_0) \\ \sqrt{n}(\hat\beta_n-\beta_0)\\
n(\hat\zeta_n-\zeta_0)\end{array}\right) \leadsto sargmax_{h\in\mathbb{R}^{2p+1}} U(h).$$
\end{corollary}

\section{Large sample properties of the bootstrap procedures}\label{SectionCheck}

In this section we use the results from the previous section to prove the (in)-consistency of different bootstrap methods introduced in Section \ref{sectionbs}.
In Section \ref{inconsistbs}, we argue that the classical bootstrap method (Method \ref{methodclassic}) and the conditional methods (Methods \ref{methodcond1} and \ref{methodcond2}) are inconsistent. 
 In  Section \ref{consistbs}, we prove the consistency of the smooth bootstrap (Methods \ref{methodsmooth1} and \ref{methodsmooth2}) and the $m$-out-of-$n$ bootstrap (Method \ref{methodmofn}). 
 
 Recall the notation and definitions in the beginning of Section \ref{SectionModel}. In particular, note that we have i.i.d. random vectors $(\tilde T_i, \delta_i, Z_i)_{i=1}^\infty$ from \eqref{cox}. 
Let  ${\cal X}$ be the $\sigma$-algebra generated by the sequence $(\tilde T_i, \delta_i, Z_i)_{i=1}^\infty.$
For a metric space $(\mathtt{X}, d)$, consider $\mathtt{X}$-valued random elements $(V_n)_{n=1}^\infty$ and $V$ defined on the probability space $(\Omega, {\cal A}, \mathbf{P})$. 
We say that $V_n$ converges conditionally in probability to $V$, in probability, if for any given 
$\epsilon>0$
$$ \mathbf{P}(d(V_n,V)>\epsilon ~|~ \mathcal{X})\xrightarrow[]{\mathbf{P}} 0,$$
and we write $V_n\xrightarrow[\mathbf{P}]{\mathbf{P}_{\mathcal{X}}} V.$

\subsection{Inconsistent bootstrap methods}\label{inconsistbs}
\subsubsection{Classical bootstrap} \label{sectionbs1}
Consider the classical bootstrap Method \ref{methodclassic} introduced in Section \ref{sectionbs}. We  set  $m_n=n$ and $\mathbb{Q}_{n,j}, j=1,\cdots,n,$ to be the ED of the data $\{(\tilde T_{i}, \delta_{i}, Z_{i}): i=1,\cdots,n\}$. This implies that for $k=0,1,2,$
\begin{eqnarray}\label{Scls}
s_{n,k}(t;\gamma)&=&
 \mathbb{Q}_n \biggr(\frac{1}{m_n}\sum_{i=1}^{m_n} Y_{n,i}(t)Z_{n,i}^{\otimes k}(t)\exp(\gamma' Z_{n,i}(t))\biggr)
 \notag\\
&=&
\frac{1}{n}\sum_{i=1}^{n}Y_{i}(t)Z_{i}^{\otimes k}(t)\exp(\gamma' Z_{i}(t)),\\
\label{Acls}
A_{n,k}(t)
 &=&\mathbb{Q}_n\biggr(\frac{1}{m_n}\sum_{i=1}^{m_n} \int_0^t Z_{n,i}^{\otimes k}(s) dN_{n,i}(s)\biggr)\notag\\
&=&\frac{1}{n}\sum_{i=1}^{n} \int_0^t Z_{i}^{\otimes k}(s) dN_{i}(s).
\end{eqnarray}
Therefore, $\theta_n = \hat\theta_n$ and condition A1 holds. Apply Theorem \ref{theorem1} and we have that the bootstrap estimator $\theta^*_n$ converges conditionally in probability to the true value $\theta_0$, in probability. 
\begin{proposition}
For  Method \ref{methodclassic}, 
$\theta^*_n \xrightarrow[\mathbf{P}]{\mathbf{P}_{\cal X}}\theta_0.$ 
\end{proposition}

As for the weak convergence, 
we  show in Lemma \ref{lemmancc} that condition A3 does not hold. 
Hence, Theorem \ref{theorem3} is not applicable in this case.  

\begin{lemma}\label{lemmancc}
 For Method \ref{methodclassic},
 there is  $h_0>0$ such that for any $h>h_0$, the sequences 
\begin{equation}\label{1e}
\biggr\{\sum_{i=1}^n \int_{\hat\zeta_n}^{\hat\zeta_n +\frac{h}{n}} dN_{i}\left(s\right)\biggr\}_{n=1}^\infty ~~and~~
\biggr\{\sum_{i=1}^n \int_{\hat\zeta_n- \frac{h}{n}}^{\hat\zeta_n} dN_{i}\left(s\right)\biggr\}_{n=1}^\infty  
\end{equation} 
 do not converge in probability. Furthermore, 
\begin{equation}\label{2e}
\biggr\{ \sum_{i=1}^{n} 
\int_{\hat\zeta_n}^{\hat\zeta_n+h/n} 
\phi_i(s) dN_{i}(s)
\biggr\}_{n=1}^\infty ~~and~~
\biggr\{\sum_{i=1}^n \int_{\hat\zeta_n- \frac{h}{n}}^{\hat\zeta_n} \phi_i(s) dN_{i}\left(s\right)\biggr\}_{n=1}^\infty,
\end{equation}
where $\phi_i(s):= e^{\imath t\left( (\hat\alpha_n-\hat\beta_n)'Z_{i}(s) - 
\log{r_{n}(s;\hat\alpha_n,\hat\beta_n)} \right)} -1$,    do not converge in probability.
\end{lemma}

The following theorem shows that, conditional on the data,  $(U^*_n)_{n=1}^\infty$ does not have any weak limit in probability.  Consider the Skorohod space ${\cal D}_{\Theta}$  with compact set $\Theta\subset \mathbb{R}^{2p+1}$. We say that $(U^*_n)_{n=1}^\infty$ has no weak limit in probability in ${\cal D}_{\Theta}$ if there is no probability measure $\mu$ defined on  ${\cal D}_{\Theta}$ such that $\rho(\mu_n,\mu)\xrightarrow[]{\mathbf{P}} 0,$ 
where $\mu_n$ is the conditional distribution of $U^*_n$ given $\mathcal{X}$, and $\rho$ is a metric metrizing weak convergence on ${\cal D}_{\Theta}$.

\begin{theorem}\label{theoremncc}
There is a compact set $\Theta\in \mathbb{R}^{2p+1}$ such that, conditional on the data, $U^*_n$ does not have a weak limit in probability in ${\cal D}_\Theta$.
\end{theorem}

\begin{proof}[Proof of Theorem \ref{theoremncc}]
It suffices to show that there is some $h>0$ such that, conditional on the data,  $U^*_n(\mathbf 0,\mathbf 0,h)$ does not have a weak limit in probability. In this case, 
$$U^*_n(\mathbf 0,\mathbf 0,h) 
= \sum_{i=1}^{n} \int_{\hat\zeta_n}^{\hat\zeta_n+h/n} \big( (\hat\alpha_n-\hat\beta_n)'Z_{n,i}(s) - 
\log{R_{n}(s;\hat\alpha_n,\hat\beta_n)}\big) dN_{n,i}(s).$$
Consider the conditional characteristic function of $U^*_n(\mathbf 0,\mathbf 0,h)$ given ${\cal X}$. A similar argument as in the proof of Lemma \ref{lemma3} implies that 
\begin{eqnarray*}
&&\mathbf{E} [e^{\imath tU^*_n(\mathbf 0,\mathbf 0,h)} |{\cal X}] =(1+o_{\mathbf{P}}(1))
\exp\Big(\sum_{i=1}^{n} 
\int_{\hat\zeta_n}^{\hat\zeta_n+h/n} \phi_i(s)dN_{i}(s)\Big),
\end{eqnarray*}
where $\phi_i(s)$ is defined as in Lemma \ref{lemmancc}.
Then Lemma \ref{lemmancc} implies the desired conclusion.
\end{proof}

The result that  $U^*_n$  does not have any weak limit in probability makes the existence of a weak limit for  $n(\zeta^*_n-\hat\zeta_n)$ very unlikely; see \eqref{argmaxu}. But a complete proof of the non-existence may be complicated due to the non-linearity of the smallest argmax functional. 
For this reason, theoretically we do not  pursue this problem any further, and we will use simulation results to illustrate the inconsistency in Section \ref{SectionSim}.

\subsubsection{Conditional bootstrap} \label{sectionbs2}
For  Methods \ref{methodcond1} and \ref{methodcond2} in Section \ref{sectionbs}, we consider that $m_n=n$, $\mathbb{Q}_{n,i}(Z_{n,i}=Z_i)=1$,
and the cumulative hazard function of $T$ takes the form $\Lambda_{n,0} = \hat\Lambda^b_{n,0}$, where  $\hat\Lambda^b_{n,0}$ is the Breslow estimator as defined in Method \ref{methodcond1}.  Therefore, for $k=0,1,2$,
\begin{eqnarray*}
s_{n,k}(t;\gamma) 
= \frac{1}{n}\sum_{i=1}^{n}\mathbb{Q}_{n,i}(Y_{n,i}(t)|Z_i)Z_{i}^{\otimes k}(t)\exp(\gamma' Z_{i}(t))\end{eqnarray*}
and 
\begin{eqnarray*}
A_{n,k}(t) 
= \int_0^t s_{n,k}(s;\hat\alpha_n 1_{s\leq \hat\zeta_n} + \hat\beta_n 1_{s> \hat\zeta_n} ) d\hat\Lambda^b_{n,0}(s).\end{eqnarray*}
Thus $\theta_n = \hat\theta_n$. 

A uniformly consistent estimator of $G$ is usually needed in conditional bootstrap methods for the Cox model. We assume that
 \begin{eqnarray}\label{CCensor}
\sup_{t\in [0,\tau], z\in {\cal V}}|\hat G_{n}(t|z) - G(t|z) |\xrightarrow[]{\mathbf{P}} 0,
\end{eqnarray}
where ${\cal V}$ is the set of all possible sample paths of covariate $Z$. 
Note that $\hat G_n$ can be taken as the Kaplan-Meier estimator when $C_i$'s are i.i.d. or $Z_i's$ are time-independent and categorical; see also \cite{Beran81} for a class of nonparametric estimates of the conditional distribution.  For more general  time-dependent covariates $Z$, it is hard, if not impossible, to obtain a consistent estimator of $G$ without further model assumption. In the literature, a common approach is to assume that the censoring time follows the Cox model \eqref{eq1}, in which case a consistent estimator of $G$ can be constructed based on the usual Breslow estimator \citep{MR751780}. In this paper we assume that \eqref{CCensor} holds and do not go into the problem of estimating $G$ any further.

Under the setup in Section \ref{SectionModel}, it is known that $\sup_{t \in [0,\tau]}|\Lambda_{n,0}(t) - \Lambda_0(t)| \xrightarrow[]{\mathbf{P}} 0$ \citep{ABGK}. 
Together with \eqref{CCensor}, this implies condition A1.
Apply Theorem \ref{theorem1} and we have the following convergence result.
\begin{proposition}
For  Methods \ref{methodcond1} and \ref{methodcond2}, if \eqref{CCensor} holds, then
$\theta^*_n \xrightarrow[\mathbf{P}]{\mathbf{P}_{\cal X}}\theta_0.$ 
\end{proposition} 

As with Method \ref{methodclassic},  we argue that   Methods \ref{methodcond1} and \ref{methodcond2} are also inconsistent. We start with the following lemma.
\begin{lemma}\label{lemmanc}
 For Methods \ref{methodcond1} and \ref{methodcond2},
 there is  $h_0>0$ such that for any $h>h_0$, the sequences 
$$\biggr\{n\biggr(\hat\Lambda^ b_{n,0}\Big(\hat\zeta_n +\frac{h}{n}\Big)
 -\hat\Lambda^ b_{n,0}\big(\hat\zeta_n\big)\biggr)\biggr\}_{n=1}^\infty ~~and~~
\biggr\{n\biggr(\hat\Lambda^ b_{n,0}\big(\hat\zeta_n \big)
 -\hat\Lambda^ b_{n,0}\Big(\hat\zeta_n-\frac{h}{n}\Big)\biggr)\biggr\}_{n=1}^\infty $$ 
 do not converge in probability.  
\end{lemma}

\begin{proof}[Proof of Lemma \ref{lemmanc}]
We only need to show that the first sequence does not converge in probability. For $h>0$,
\begin{eqnarray*}
&&n\biggr(\hat\Lambda^ b_{n,0}\Big(\hat\zeta_n +\frac{h}{n}\Big)
 -\hat\Lambda^ b_{n,0}\big(\hat\zeta_n\big)\biggr)\\
&=& n\int_{\hat\zeta_n }^{\hat\zeta_n +\frac{h}{n}} 
\Big(\sum_{j=1}^n Y_j(s)e^{\hat\alpha_n' Z_j(s) 1_{s\leq \hat\zeta_n} 
 +\hat \beta'_n Z_j(s)1_{s> \hat\zeta_n}}\Big)^{-1}d\Big(\sum_{i=1}^n N_i(s)\Big)\\
&=& (1+o_{\mathbf{P}}(1)) \int_{\hat\zeta_n }^{\hat\zeta_n +\frac{h}{n}} 
s_{0}(s;\beta_0)^{-1}d\Big(\sum_{i=1}^n N_i(s)\Big)\\
&=& (1+o_{\mathbf{P}}(1)) s_{0}(\zeta_0;\beta_0)^{-1}
\sum_{i=1}^n  \int_{\hat\zeta_n}^{\hat\zeta_n +\frac{h}{n}} dN_i(s).
\end{eqnarray*}
Thus, it suffices to show that $\sum_{i=1}^n  \int_{\hat\zeta_n}^{\hat\zeta_n +\frac{h}{n}} dN_i(s)$ does not converge in probability. Apply Lemma \ref{lemmancc} and we have the desired conclusion.
\end{proof}

Based on Lemma \ref{lemmanc}, we further show that, conditional on the data, the sequence $\{U^*_n\}_{n=1}^{\infty}$ does not have a weak limit in probability.

\begin{theorem}\label{theoremnc}
There is a compact set $\Theta\in \mathbb{R}^{2p+1}$ such that, conditional on the data, $U^*_n$ does not have a weak limit in probability in ${\cal D}_\Theta$.
\end{theorem}

\begin{proof}[Proof of Theorem \ref{theoremnc}]
For $h>0$, consider the conditional characteristic function of $U^*_n(\mathbf 0,\mathbf 0,h)$ given ${\cal X}$. A similar argument as in the proof of Lemma \ref{lemma3} implies that 
\begin{align*}
\mathbf{E}[e^{\imath tU^*_n(\mathbf 0,\mathbf 0,h)}|{{\cal X}}]  =~&
(1+o_{\mathbf{P}}(1))
 \exp\biggr\{n\Big(\hat\Lambda^ b_{n,0}\Big(\hat\zeta_n +\frac{h}{n}\Big)
 -\hat\Lambda^ b_{n,0}\left(\hat\zeta_n\right)\Big)\\
&~~~~~~\times \Big( e^{- \log{r(\zeta_0;\alpha_0,\beta_0)} }
  ~s_{0}(\zeta_0; \imath t (\alpha_0-\beta_0)+\beta_0)
-s_{0}(\zeta_0; \beta_0)\Big) \biggr\}.
\end{align*}
Hence,  Lemma \ref{lemmanc} implies the desired conclusion.
\end{proof}

\subsection{Consistent bootstrap methods}\label{consistbs}
In this section we show that the smooth bootstrap  (Methods \ref{methodsmooth1} and \ref{methodsmooth2}) and the $m$-out-of-$n$ bootstrap (Method \ref{methodmofn}) are consistent for constructing CIs for $\zeta_0$.

The results from Section \ref{SectionGeneral} can be directly applied to derive sufficient conditions on the distribution from which the bootstrap samples are generated. 
Let $\hat {\mathbb Q}_n$ be a distribution constructed from the data $\{(\tilde T_{i}, \delta_{i}, Z_{i}): i=1,\cdots,n\}$.  If conditions A1-A3 hold with 
${\mathbb Q}_n=\hat {\mathbb Q}_n$, then the weak convergence of the bootstrap estimate follows from Theorem \ref{theorem3} applied conditionally given the data.

\subsubsection{Smooth bootstrap}\label{sectionsmooth}
Consider Methods \ref{methodsmooth1} and \ref{methodsmooth2}. To prove the consistence, thanks to Theorem \ref{theorem3}, we only need to show conditions A1-A3 hold conditionally on the data with $m_n=n$ and ${\mathbb Q}_n$ the distribution of the bootstrap sample.  Recall that $\hat\lambda_{n,0}(\cdot)$ and $\hat G(\cdot|Z)$ are the estimated smooth baseline hazard rate function of $T$ and the conditional distribution of $C$ given $Z$, respectively. 
In addition to \eqref{CCensor}, we need the following convergence result:
\begin{eqnarray}\label{hazardc}
\sup_{t\in [0,\tau]}|\hat\lambda_{n,0}(t) - \lambda_0(t)|\xrightarrow[]{\mathbf{P}} 0.
\end{eqnarray}
Note that \eqref{hazardc} is fulfilled if $\hat\lambda_{n,0}$ is the usual kernel estimator \citep{Wells94}. 

Similarly as in Section \ref{sectionbs2}, we have 
$\theta_n=\hat\theta_n$, 
$s_{n,k}(t;\gamma) \xrightarrow[]{\mathbf{P}}  s_{k}(t;\gamma)$,
and $A_{n,k}(t)  \xrightarrow[]{\mathbf{P}} A_{k}(t).$
In addition,  
$$\int_{\hat\zeta_n}^{\hat\zeta_n+h} dA_{n,1}(s)
- \int_{\hat\zeta_n}^{\hat\zeta_n+h}
{\bar z_{n}(s; \hat\alpha_n 1_{s\leq \hat\zeta_n} + \hat\beta_n 1_{s> \hat\zeta_n})}
 dA_{n,0}(s)
= 0,
$$ and for any $t\in \mathbb{R}$, $h_1<h_2$ and $0\notin (h_1,h_2)$,
\begin{eqnarray*}
&& {\mathbb Q}_n\biggr(\sum_{i=1}^{n}
\int_{\hat\zeta_n+h_1/n}^{\hat\zeta_n+h_2/n} 
e^{\imath t (\hat\alpha_n-\hat\beta_n)'Z_{n,k}(s)} dN_{n,i}(s)\biggr)\\
&=&
n\int_{\hat\zeta_n+h_1/n}^{\hat\zeta_n+h_2/n} 
s_{0}\big(s;\gamma_n+\imath t(\hat\alpha_n-\hat\beta_n)\big)\hat\lambda_{0}(s)ds \\
&\xrightarrow[]{\mathbf{P}}&
s_{0}\big(\zeta_0;\gamma_0+\imath t(\alpha_0-\beta_0)\big)\lambda_0(\zeta_0)(h_2-h_1),
\end{eqnarray*}
where $\gamma_n = \alpha_n1_{h_2\leq 0}+\beta_n1_{0\leq h_1}$ and 
$\gamma_0 = \alpha_01_{h_2\leq 0}+\beta_01_{0\leq h_1}.$ Therefore, A1-A3 hold and  Theorems \ref{theorem1} and \ref{theorem3} give the weak consistency result.
\begin{proposition}
For  Methods \ref{methodsmooth1} and \ref{methodsmooth2}, if \eqref{CCensor} and \eqref{hazardc} hold, then $\theta^*_n\xrightarrow[\mathbf{P}]{\mathbf{P}_{\cal X}} \theta_0$ and conditional on the data,
\begin{equation}\label{propsmooth}
\left(\begin{array}{c} \sqrt{n}(\alpha_n^*-\hat\alpha_n) \\ \sqrt{n}(\beta_n^*-\hat\beta_n)\\
n(\zeta_n^*-\hat\zeta_n)\end{array}\right) \leadsto sargmax_{h\in\mathbb{R}^{2p+1}} U(h).
\end{equation}
\end{proposition}

\subsubsection{$m$-out-of-$n$ bootstrap}
Consider the $m$-out-of-$n$ bootstrap (Method \ref{methodmofn}). 
We will again use the consistency results established in Section \ref{SectionGeneral}. 
We set $m_n\rightarrow \infty$ and $m_n/n\rightarrow 0$ as $n\rightarrow\infty$.
Similar to the classical bootstrap, $\mathbb{Q}_{n,j}$, $j=1,\cdots,n$, is the ED of the data $\{(\tilde T_{i}, \delta_{i}, Z_{i}): i=1,\cdots,n\}$,  and \eqref{Scls} and \eqref{Acls} hold.
Therefore $\theta_n = \hat\theta_n$ and condition A1 holds. Consider the first equation in A2 and we have 
\begin{align*}
&
\sup_{|h|> h_n/m_n}\frac{1}{h}\biggr|\int_{\hat\zeta_n}^{\hat\zeta_n+h} dA_{n,1}(s)
- \int_{\hat\zeta_n}^{\hat\zeta_n+h}
{\bar z_{n}(s; \hat\alpha_n 1_{s\leq \hat\zeta_n} + \hat\beta_n 1_{s> \hat\zeta_n})}
 dA_{n,0}(s)\biggr|\\
=&
\sup_{|h|> h_n/m_n}
\frac{1}{nh}\biggr|\sum_{i=1}^{n}\int_{\hat\zeta_n}^{\hat\zeta_n+h} \big[Z_i
- {\bar z_{n}(s; \hat\alpha_n 1_{s\leq \hat\zeta_n} + \hat\beta_n 1_{s> \hat\zeta_n})}
\big] dN_{i}(s)\biggr|\\
\xrightarrow[]{\mathbf{P}}&\quad 0.
\end{align*} 
As to condition A3, for any $t\in \mathbb{R}$, $h_1<h_2$ and $0\notin (h_1,h_2)$,
\begin{eqnarray*}
&& {\mathbb Q}_n\biggr(\sum_{i=1}^{m_n}
\int_{\hat\zeta_n+h_1/m_n}^{\hat\zeta_n+h_2/m_n} 
e^{\imath t (\hat\alpha_n-\hat\beta_n)'Z_{n,i}(s)} dN_{n,i}(s)\biggr)\\
&=& 
\frac{m_n}{n}\sum_{i=1}^{n}\int_{\hat\zeta_n+h_1/m_n}^{\hat\zeta_n+h_2/m_n} 
e^{\imath t (\hat\alpha_n-\hat\beta_n)'Z_{i}(s)} dN_{i}(s) \\
&\xrightarrow[]{\mathbf{P}}&
s_{0}\left(\zeta_0;\gamma_0+\imath t(\alpha_0-\beta_0)\right)\lambda_0(\zeta_0)(h_2-h_1),
\end{eqnarray*}
where $\gamma_0 = \alpha_01_{h_2\leq 0}+\beta_01_{0\leq h_1}.$
Therefore, A1-A3 hold and we have the following proposition. 
\begin{proposition}
 For the $m$-out-of-$n$ bootstrap method, if $m_n\rightarrow \infty$ and $m_n/n\rightarrow 0$ as $n\rightarrow\infty$,
then $\theta^*_n\xrightarrow[\mathbf{P}]{\mathbf{P}_{\cal X}} \theta_0$ and conditional on the data
$$\left(\begin{array}{c} \sqrt{m_n}(\alpha_n^*-\hat\alpha_n) \\ \sqrt{m_n}(\beta_n^*-\hat\beta_n)\\
m_n(\zeta_n^*-\hat\zeta_n)\end{array}\right) \leadsto sargmax_{h\in\mathbb{R}^{2p+1}} U(h).
$$
\end{proposition}

\section{Simulation}\label{SectionSim}
In this section we compare the finite sample performance of the different bootstrap schemes introduced in Section \ref{sectionbs}. 
We consider a single covariate $Z$ which has a Bernoulli distribution with parameter $0.5$. That is, a subject is equally likely to be assigned to the control group ($Z=0$) and the treatment group ($Z=1$). 
The model parameter values are set at $\alpha_0=0,\beta_0=-1.5,$ and $\zeta_0=1.$ 
The baseline hazard rate is assumed constant and taken as $\lambda_0(t)=0.5$. Note that, at $\zeta_0=1$, the cumulative mortality for the control group is $1-\exp(-0.5)=39\%$. The censoring times are chosen to be independent and follow an exponential distribution with rate parameter $0.1$ and truncated at $\tau=4$.
This results in a censoring rate of about $36\%$.
 Figure \ref{Figure1} gives the Kaplan-Meier curves of a simulated sample of size $n=1000$, which clearly shows the  lag feature around the change-point time $\zeta_0=1$. 

\begin{figure}[h]
\centering
    \includegraphics[width=0.45\textwidth]{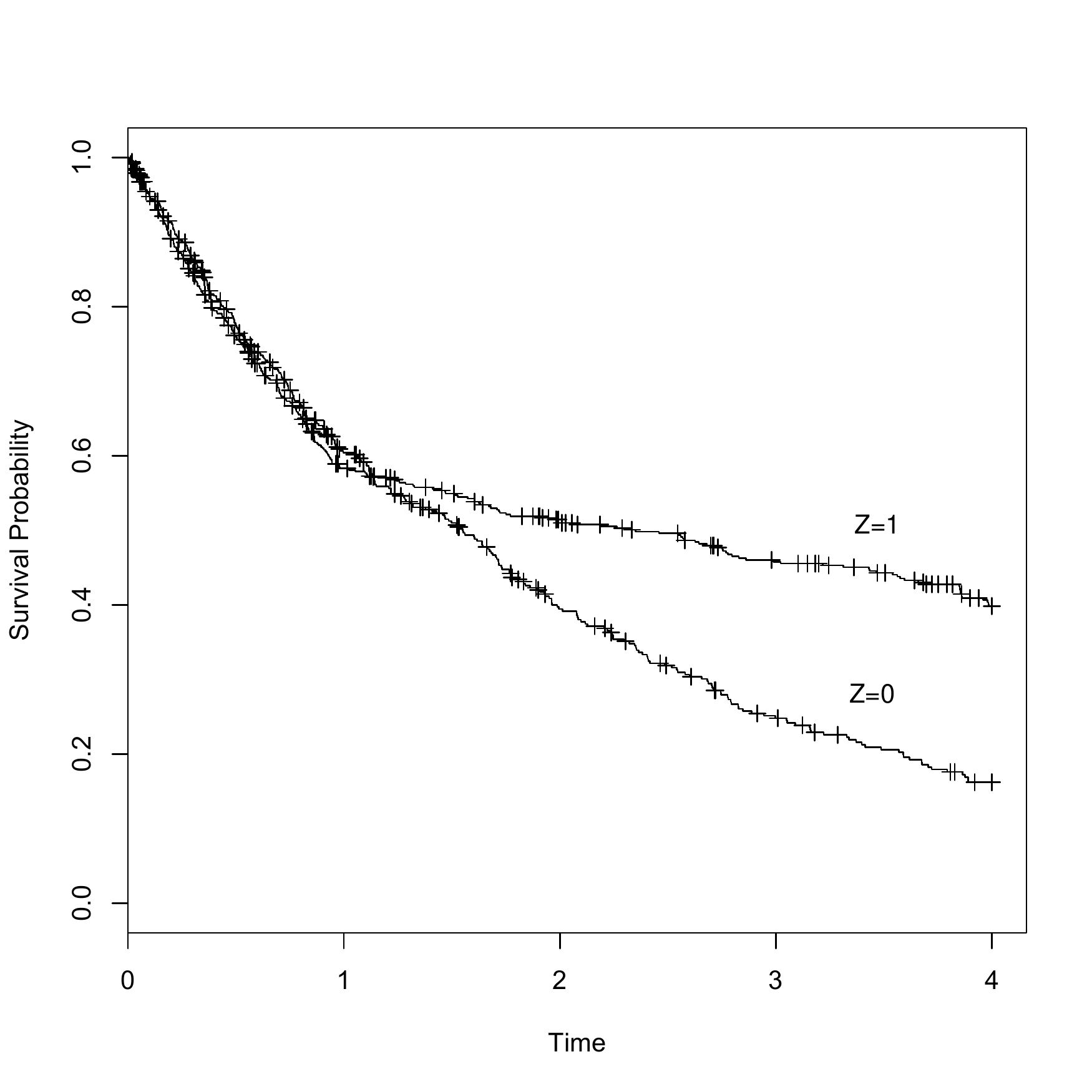}
  \caption{Kaplan-Meier curves of a simulated sample.}
  \label{Figure1}
\end{figure} 

We consider 1000 random samples of sample sizes $n=200, 500,1000$. For each simulated sample and for each bootstrap method,  $1000$ bootstrap replicates are generated to approximate the bootstrap distribution. The  conditional censoring distribution estimator $\hat G(\cdot |Z)$ is taken as the Kaplan-Meier estimators for each group ($Z=0$ and $Z=1$). For the smooth bootstrap, we use a kernel density estimator based on the Gaussian kernel and choose the so-called ``normal-reference rule'' \citep{Scott1992}. For the $m$-out-of-$n$ bootstrap, we try three different choices of $m_n$: $n^{4/5}$, $n^{9/10}$ and $n^{14/15}$. To reduce the computation complexity, we restrict $\zeta\in[0.5,1.5]$ when calculating the MPLE of $\zeta_0$.

Table \ref{t1} provides the simulation results of coverage proportions and average lengths of nominal $95\%$ CIs for $\zeta_0$ that are estimated using different bootstrap methods.
The first column (``Smooth'') gives the results of smooth bootstrap Method \ref{methodsmooth1}, the second column (``Classical'') corresponds to classical bootstrap Method \ref{methodclassic}, the third column (``Conditional'') corresponds to conditional bootstrap Method \ref{methodcond1}, and the last three columns correspond to the $m$-out-of-$n$ bootstrap with different choices of $m_n$. 
The results of Methods \ref{methodcond2} and \ref{methodsmooth2} are similar to those of Method \ref{methodcond1} and Method \ref{methodsmooth1} and therefore are not presented. 

\begin{table}[h]
\begin{center}
\begin{tabular}{|c|c|c|c|c|c|c|c|c|}
\hline
&~&Smooth&Classical&Conditional&$m_n=n^{4/5}$&$n^{9/10}$&$n^{14/15}$\\
\hline $n=300$
& Coverage & 0.96 & 0.88 & 0.87 & 0.97 &0.93 &0.91\\
 & Length    & 0.64 & 0.56 & 0.54 & 0.80 &0.69 &0.64\\
\hline $500$
& Coverage & 0.96 & 0.89 & 0.89 & 0.98 &0.96 &0.94\\
 &  Length   &0.46  & 0.48 & 0.46 & 0.77 &0.62 &0.58\\
\hline $1000$
& Coverage & 0.95 & 0.89 & 0.88 & 0.99 &0.97 &0.94\\
 &  Length  & 0.24 & 0.30 & 0.29 & 0.69 &0.48 &0.42\\
\hline
\end{tabular}
 \end{center}
\caption{The estimated coverage rates and average lengths of nominal 95\% CIs for $\zeta_0$.}
\label{t1}
\end{table}

We can see from Table \ref{t1} that the smooth bootstrap outperforms all the others in terms of coverage rate and average length. The $m$-out-of-$n$ bootstrap also performs reasonably well, but the average length is bigger than that of the smooth bootstrap. This may be due to the fact that the $m$-out-of-$n$ bootstrap method converges at rate $m_n^{-1}$ instead of $n^{-1}$.   Table \ref{t1}  also shows that the commonly used bootstrap Methods \ref{methodclassic} and \ref{methodcond1} provide under-coverage, which indicates their inconsistency.

To further illustrate the performance of different bootstrap methods, we compare the histograms of the distribution of $n(\hat\zeta_n-\zeta_0)$, obtained from 1000 random samples of sample size 1000, and its bootstrap estimates from a single sample. All bootstrap estimates are based on 1000 bootstrap replicates. It is clearly shown in Figure \ref{Figure2}  that the smooth bootstrap (top right panel) provides the best approximation to the actual distribution obtained from 1000 random samples (top left panel).  
\begin{figure}[h]
\centering
    \includegraphics[width=0.75\textwidth]{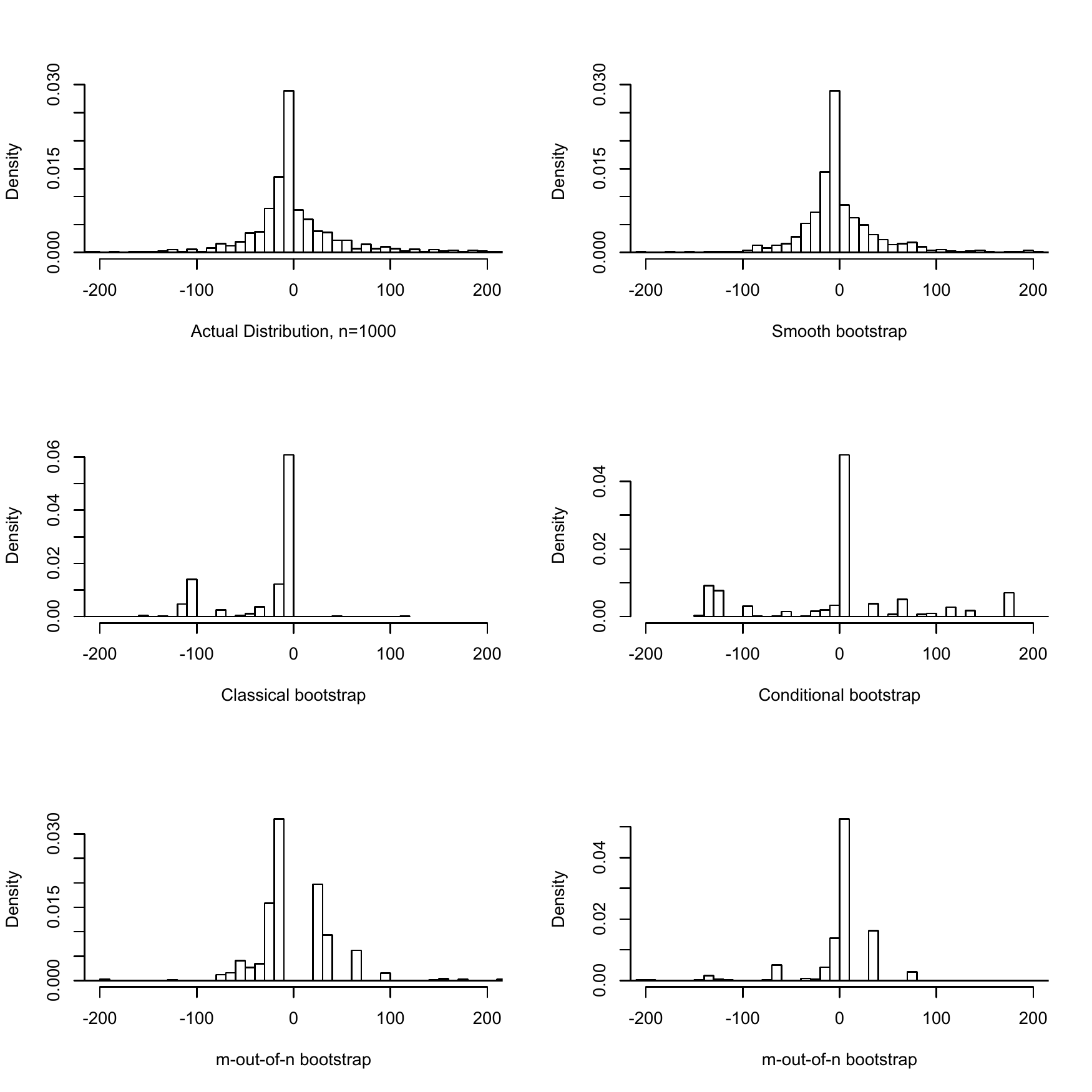}
  \caption{Histograms of the distribution of $n(\hat\zeta_n-\zeta_0)$ and its bootstrap estimates. The top left panel shows the distribution of $n(\hat\zeta_n-\zeta_0)$ obtained from 1000 random samples; the top right panel shows the distribution of $n(\hat\zeta^*_n-\hat\zeta_n)$ for the smooth bootstrap (Method \ref{methodsmooth1}), the middle left the classical bootstrap (Method \ref{methodclassic}), the middle right  the conditional bootstrap (Method \ref{methodcond1}). The bottom left panel shows the distribution of $m_n(\hat\zeta^*_n-\hat\zeta_n)$ with $m_n=n^{9/10}$, and the bottom right  for $m_n=n^{14/15}$.}
  \label{Figure2}
\end{figure}

\section{Proof of Theorems}\label{SectionProof}
This section contains proofs of Theorems \ref{theorem1}, \ref{theorem2} and \ref{theorem3}.
\subsection{Proof of Theorem \ref{theorem1}}

We first show that $\theta_n\rightarrow \theta_0$. Write  $\Theta =\Theta_\alpha\times \Theta_\beta\times [0,\tau]$.
By the definition, $\theta_n=(\alpha_n',\beta_n',\zeta_n)'$ is the smallest maximizer of
\begin{align*} 
X_0(\theta):=X_0(\alpha,\beta,\zeta)
:=&
 \int_0^\tau \big((\alpha-\alpha_0) 1_{s\leq \zeta} + (\beta-\beta_0) 1_{s> \zeta} \big)' d A_{n,1}(s)\notag\\
&-  \int_0^\tau\log \frac{s_{n,0} (s; \alpha 1_{s\leq \zeta} + \beta 1_{s> \zeta})}
{s_{n,0} (s; \alpha_0 1_{s\leq \zeta_0} + \beta_0 1_{s> \zeta_0})}
dA_{n,0}(s).
\end{align*}
 Thus, $X_0(\theta_n)\geq 0$. By condition A1, we have   
\begin{eqnarray*} 
|X_0(\theta_n)-X(\theta_n)|\leq \sup_{\theta\in\Theta}|X_0(\theta)-X(\theta)|\rightarrow 0,
\end{eqnarray*}
where 
\begin{align}\label{X}
X(\theta) =  \int_0^\tau \left((\alpha-\alpha_0)'s_1(s;\alpha_0) - s_{0}(s; \alpha_0)
\log {r(s; \alpha,\alpha_0)}\right)1_{s\leq \zeta\wedge\zeta_0} d\Lambda_0(s)\notag\\
+  \int_0^\tau \left((\alpha- \beta_0)'s_1(s;\beta_0) -s_{0}(s; \beta_0) 
\log {r(s; \alpha, \beta_0)}\right)1_{ \zeta_0<s\leq\zeta} d\Lambda_0(s)\notag\\
+ \int_0^\tau \left((\beta-\alpha_0)'s_1(s;\alpha_0)  - s_0(s;\alpha_0)
\log {r(s; \beta, \alpha_0)}\right)1_{\zeta<s \leq \zeta_0}d\Lambda_0(s)\notag\\
+  \int_0^\tau \left((\beta-\beta_0)'s_1(s;\beta_0)  - s_{0}(s; \beta_0)
\log {r(s; \beta, \beta_0)}\right)1_{s> \zeta\vee\zeta_0} d\Lambda_0(s).
\end{align}
Then by the continuous mapping theorem, the conclusion follows from the fact that $\theta_0$ is the unique maximizer of $X(\theta)$ and that $X(\theta_0)=0$.   

\smallskip
We now show the consistency of $\hat\theta^*_n$. 
For $\gamma_n\in\{\alpha_n,\beta_n\}$ and $\gamma\in\Theta_\alpha\cup\Theta_\beta$, let
\begin{eqnarray*}
w_n(t;\gamma) := \int_0^t (\gamma-\gamma_n)' dM_{n,1}(s) -  \int_0^t \log {r_{n}(s; \gamma, \gamma_n)}dM_{n,0}(s),
\end{eqnarray*}
where 
\begin{align}\label{labelM}
M_{n,k}(t) =\frac{1}{m_n}\sum_{i=1}^{m_n} \int_0^t Z_{n,i}^{\otimes k}(s) dN_{n,i}(s)-A_{n,k}(t), ~\mbox{ for } k =0,1.
\end{align}
For any $\epsilon_1>0$, we have
\begin{align*}
\mathbb{Q}_{n}\Big( \sup_{t\in [0,\tau]} |w_n(t;\gamma)|  \geq 2\epsilon_1\Big)
&\leq  \mathbb{Q}_{n}\Big( \sup_{t\in [0,\tau]} 
\Big| \int_0^t (\gamma-\gamma_n)' dM_{n,1}(s)\Big|^2  \geq \epsilon^2_1\Big)\notag\\
&+\mathbb{Q}_{n}\Big( \sup_{t\in [0,\tau]} 
\Big|\int_0^t \log {r_{n}(s; \gamma, \gamma_n)}dM_{n,0}(s)\Big|^2  \geq \epsilon^2_1\Big).\notag
\end{align*}
Then by Lenglart's inequality for c\'adl\'ag processes  \citep[p.35]{Jacob}, we have that for  
$\epsilon_2>0$, there exists a  constant $B >0$ such that
\begin{eqnarray}\label{inequality1}
&& \mathbb{Q}_{n}\Big( \sup_{t\in [0,\tau]} |w_n(t;\gamma)|  \geq 2\epsilon_1\Big)\notag\\
  & \leq& 2\Big(\frac{\epsilon_2}{\epsilon_1} + \frac{B}{\epsilon^2_1m_n}\Big)
 +   \mathbb{Q}_{n}\Big(\frac{1}{m_n^2}\sum_{i=1}^{m_n} 
 \int_0^t  \big((\gamma-\gamma_n)' Z_{n,i}(s)\big)^{2} dN_{n,i}(s)>\epsilon_2\Big)\notag\\
&& + 
  \mathbb{Q}_{n}\biggr(\frac{1}{m_n^2} \sum_{i=1}^{m_n} 
 \int_0^t  \big(\log {r_{n}(s; \gamma,\gamma_n)}\big)^{2} dN_{n,i}(s)>\epsilon_2
 \biggr)\nonumber\\ 
 &\leq & 2\Big(\frac{\epsilon_2}{\epsilon_1^2} + \frac{B}{\epsilon^2_1m_n}\Big)
+ \frac{2B}{\epsilon_2m_n^2},
\end{eqnarray}
where the last inequality follows from  Chebyshev inequality. Since $\epsilon_1$ and $\epsilon_2$ are arbitrary, it follows that
\begin{equation}\label{superw}
\sup_{t\in [0,\tau]} |w_n(t;\gamma)|  \xrightarrow[]{\mathbf{P}} 0, ~\hbox{ for }\gamma\in\Theta_\alpha\cup\Theta_\beta.
\end{equation}
On the other hand, by Lemma \ref{lemma1} in the Appendix, we have 
\begin{equation}\label{convergew}
\sup_{t\in [0,\tau]} \biggr|\frac{1}{m_n} \sum_{i=1}^{m_n}
 \int_0^t \left(\log{R_{n}(s; \gamma,\gamma_n)} 
-\log {r_{n}(s; \gamma, \gamma_n)} \right)
dN_{n,i}(s)\biggr|  \rightarrow 0.
\end{equation}
Thus, \eqref{superw} and \eqref{convergew} imply that for  $\theta=(\alpha',\beta',\zeta)'$ with  $\alpha\in\Theta_\alpha$ and $\beta\in\Theta_\beta$,
$$ \sup_{\zeta\in [0,\tau]} |X_n(\theta) - X^*_{n}(\theta)  |  \xrightarrow[]{\mathbf{P}} 0,$$
where $X_n(\theta) = {m_n}^{-1}( l^*_n(\theta) - l^*_n(\theta_n))$ and
\begin{align*}
X^*_{n}(\theta) =\,&
\mathbb{Q}_n\biggr( \frac{1}{m_n} \sum_{i=1}^{m_n} \int_0^\tau \left((\alpha-\alpha_n)'Z_{n,i}(s) -
 \log {r_{n}(s; \alpha, \alpha_n)}\right)1_{s\leq \zeta\wedge\zeta_n} dN_{n,i}(s)\biggr)\\
&+ \mathbb{Q}_n\biggr(\frac{1}{m_n} \sum_{i=1}^{m_n} \int_0^\tau \left((\alpha-\beta_n)'Z_{n,i}(s) -
 \log {r_{n}(s; \alpha, \beta_n)}\right)1_{\zeta_n<s\leq\zeta} dN_{n,i}(s)\biggr)\\
&+\mathbb{Q}_n\biggr(\frac{1}{m_n}  \sum_{i=1}^{m_n} \int_0^\tau \left((\beta-\alpha_n)'Z_{n,i}(s) 
- \log {r_{n}(s; \beta, \alpha_n)}\right)1_{\zeta<s \leq\zeta_n} dN_{n,i}(s)\biggr)\\
&+\mathbb{Q}_n\biggr(\frac{1}{m_n}  \sum_{i=1}^{m_n} \int_0^\tau \left((\beta-\beta_n)'Z_{n,i}(s)
- \log {r_{n}(s; \beta, \beta_n)}\right)1_{s> \zeta\vee\zeta_n} dN_{n,i}(s)\biggr).
\end{align*}
A similar argument as in the proof of Theorem II.1 in \cite{AndersonGill} implies that the convergence of $ |X_n(\theta) - X^*_{n}(\theta)  | $ is uniform in $\theta\in\Theta$.
Then by the result that $\theta_n\rightarrow\theta_0$ and condition A1, we have that 
$$\sup_{\theta\in\Theta}|X_n(\theta) - X(\theta)|\xrightarrow[]{\mathbf{P}} 0,$$
where $X(\theta)$ is defined as in \eqref{X}.
Apply Corollary 3.2.3 (ii) in \cite{Vaart} and we obtain the desired convergence result.

\subsection{Proof of Theorem \ref{theorem2}}
 For $\alpha=(\alpha_1,\cdots,\alpha_p)'$ and $\beta=(\beta_1,\cdots,\beta_p)'$, let  
$Conv(\alpha,\beta)$ $:= \{\gamma=(\gamma_1,\cdots,\gamma_p)': 
\beta_i\wedge\alpha_i\leq\gamma_i\leq\beta_i\vee\alpha_i , i=1,\cdots,p\}$.
By the definition of MPLE, we have $0\leq X_n(\theta^*_n)-X_n(\theta_n)$. 
Take Taylor's expansion  of $X_n(\theta^*_n)-X_n(\theta_n)$ with respect to $\alpha_n$ and $\beta_n$ and we obtain that there exist $\tilde\alpha_n\in Conv(\alpha^*_n, \alpha_n)$ and  $\tilde\beta_n \in Conv(\beta^*_n,\beta_n)$ such that 
\begin{align*}
0
\,\leq\,&~ \sqrt{m_n}|\alpha_n^*-\alpha_n| I_1+\sqrt{m_n}|\beta_n^*-\beta_n|I_2\\
& - \frac{{m_n}}{2}|\alpha_n^*-\alpha_n|^2
\Big|\frac{1}{{m_n}} \sum_{i=1}^{m_n}\int_0^\tau Q_n(s;\tilde \alpha_n)
1_{s\leq \zeta_n}dN_{n,i}(s) \Big|\\
& - \frac{{m_n}}{2}|\beta_n^*-\beta_n|^2
\Big|\frac{1}{{m_n}} \sum_{i=1}^{m_n}\int_0^\tau Q_n(s;\tilde\beta_n)1_{s> \zeta_n}dN_{n,i}(s) \Big| \\
&+{m_n}|\zeta_n-\zeta^*_n| a_n+{m_n}|\zeta_n-\zeta^*_n| b_n,
\end{align*}
where
\begin{align*}
I_1 = &~ \Big|\frac{1}{\sqrt{m_n}}
 \sum_{i=1}^{m_n} \int_0^\tau \left(Z_{n,i}(s) - 
 {\bar Z_n(s;\alpha_n)}\right)1_{s\leq \zeta_n} d N_{n,i}(s)\Big| ,\\
 I_2=&~\Big|\frac{1}{\sqrt{m_n}}\sum_{i=1}^{m_n} \int_0^\tau \left(Z_{n,i}(s) 
- {\bar Z_n(s;\beta_n)}\right)1_{s> \zeta_n} d N_{n,i}(s)\Big|,
\end{align*}
\begin{align*}
a_n=&\frac{1}{{m_n}|{\zeta_n-\zeta^*_n|}}  \sum_{i=1}^{m_n} 
\int_0^\tau\left( (\alpha_n^*-\beta_n^*)'Z_{n,i}(s) - \log{R_{n}(s;\alpha^*_n,\beta_n^*)}\right) 1_{{\zeta_n}<s\leq{\zeta^*_n}}d N_{n,i}(s),\notag\\
b_n=& \frac{1}{{m_n}|{\zeta_n-\zeta^*_n|}} \sum_{i=1}^{m_n} \int_{0}^\tau\left( (\beta_n^*-\alpha_n^*)'Z_{n,i}(s) - \log{R_{n}(s;\beta^*_n,\alpha_n^*)}\right)  1_{{\zeta^*_n}<s\leq{\zeta_n}}d N_{n,i}(s),
\end{align*}
and $Q_n$ is defined as
\begin{equation*}\label{Qn}
Q_n(s;\alpha): = \frac{S_{n,2}(s; \alpha)}{S_{n,0}(s; \alpha)} - {\bar Z_n(s; \alpha)} ^{\otimes 2}.
\end{equation*}
Lemma \ref{lemma1} implies that $Q_n(s;\alpha)$ converges uniformly to $Q(s;\alpha)$, where
$Q(s;\alpha) = {s_{2}(s; \alpha)}/{s_{0}(s; \alpha)} -{\bar z(s; \alpha)}^{\otimes 2}$ is as defined in \eqref{QQQ}.
Let  $\sigma_0(A)$ denote the smallest eigenvalue of a matrix A.   Since $\sigma_0$ is a continuous function on $\mathbb{R}^{p\times p}$, we have $\sigma_0(Q_n)$ converges  to $\sigma_0(Q)$ and therefore $\sigma_0(Q_n(s;\tilde\alpha_n))$ and $\sigma_0(Q_n(s;\tilde\beta_n))$  are positive for all large $n$. Then by the positive definiteness of $Q_n$, we obtain 
\begin{align}
0
\,\leq\,& \sqrt{m_n}|\alpha_n^*-\alpha_n|I_1+\sqrt{m_n}|\beta_n^*-\beta_n|I_2
 - \frac{{m_n}}{2}|\alpha_n^*-\alpha_n|^2 I_3  - \frac{{m_n}}{2}|\beta_n^*-\beta_n|^2 I_4\notag \\
&+{m_n}|\zeta_n-\zeta^*_n| a_n+{m_n}|\zeta_n-\zeta^*_n| b_n, \label{dif5}
\end{align}
where 
\begin{align*}
I_3& =~\frac{1}{{m_n}} \sum_{i=1}^{m_n}\int_0^\tau \sigma_0(Q_n(s;\tilde \alpha_n))
1_{s\leq \zeta_n}dN_{n,i}(s) ,\\
I_4& = ~\frac{1}{{m_n}} \sum_{i=1}^{m_n}\int_0^\tau \sigma_0(Q_n(s;\tilde\beta_n))1_{s> \zeta_n}dN_{n,i}(s).
\end{align*}

We consider the quantities in \eqref{dif5} one by one. For $I_1$, we have 
\begin{eqnarray}\label{first}
I_1&\leq&\biggr|\frac{1}{\sqrt{m_n}} \sum_{i=1}^{m_n} 
\int_0^\tau \left(Z_{n,i}(s) - {\bar z_n(s;\alpha_n)}\right)
1_{s\leq \zeta_n} d N_{n,i}(s)\biggr| \notag\\
&&+\biggr|\frac{1}{\sqrt{m_n}} \sum_{i=1}^{m_n} 
\int_0^\tau \left( {\bar z_n(s;\alpha_n)}- {\bar Z_n(s;\alpha_n)}\right)
 1_{s\leq \zeta_n} d N_{n,i}(s)\biggr| .
 \end{eqnarray}
By the definition of $\alpha_n$ and $\zeta_n$, the first term in \eqref{first} equals 
$$ \sqrt{m_n}\Big| \int_0^{\zeta_n} dM_{n,1}(s) - \int_0^{\zeta_n}{\bar z_n(s;\alpha_n)}d M_{n,0}(s)\Big|,$$
where $M_{n,k}, k=0,1,$ is defined as in \eqref{labelM}.
Then similarly as in the derivation of \eqref{inequality1}, Lenglart's inequality implies that the above quantity is $O_{\mathbf{P}}(1).$ Consider the second term in \eqref{first}.
From the proof of Lemma \ref{lemma1}, we know that $\{Y_{n,i}(t)Z_{n,i}^{\otimes k}e^{\gamma'Z_{n,i}(t)}\}$ is manageable. Then by the boundedness property of $Z$ and inequality (7.10) in page 38 of Pollard (1990), there exists constant $B>0$ such that 
\begin{eqnarray*}
\mathbf{E} \Big[\sup_{s\in[0,\tau]}\big| {\bar z_n(s;\alpha_n)}
 - {\bar Z_n(s;\alpha_n)}\big| \Big]
\leq   \frac{B}{\sqrt{m_n}} ,
\end{eqnarray*}
which implies that the second term in \eqref{first} is also $O_{\mathbf{P}}(1).$ Therefore,  
 \begin{eqnarray*}
I_1 = O_{\mathbf P}(1). 
\end{eqnarray*}
Similarly, we obtain that
$I_2 = O_{\mathbf{P}}(1).$

Consider $I_3$ and  we have that 
\begin{eqnarray*}
&& \biggr|I_3
- \int_0^{\zeta_n} \sigma_0(Q(s;\tilde \alpha_n))dA_{n,0}(s) \biggr| \\
 &\leq &
\biggr|\frac{1}{m_n} \sum_{i=1}^{m_n}
\int_0^{\zeta_n} \left(\sigma_0(Q_n(s;\tilde \alpha_n))
- \sigma_0(Q(s;\tilde \alpha_n)) \right)dN_{n,i}(s)  \biggr| \\ 
&&+ \sup_{\zeta\in [0,\tau]} \biggr| \int_0^{\zeta} \sigma_0(Q(s;\tilde \alpha_n))dM_{n,0}(s) \biggr|.
\end{eqnarray*}
The first term in the right hand side of the above display converges to 0 due to the convergence of $\sigma_0(Q_n)$. The second term also converges to 0 in probability by  Lenglart's inequality. Thus, together with condition A1 and the convergence of $\theta_n$ to $\theta_0$, we have 
\begin{eqnarray*}
 I_3=(1+o_{\mathbf{P}}(1)) \int_0^{\zeta_0} \sigma_0(Q(s; \alpha_0))
 s_0(s;\alpha_0)d\Lambda_{0}(s) .
 \end{eqnarray*}
 Similarly, 
$ I_4=(1+o_{\mathbf{P}}(1)) \int_{\zeta_0}^\tau \sigma_0(Q(s; \beta_0))
 s_0(s;\beta_0)d\Lambda_{0}(s) . $
 
Consider $a_n$ in \eqref{dif5} and 
we have that 
\begin{align}\label{approxan}
&\Big| {a_n}-\frac{1_{{\zeta_n}<{\zeta^*_n}}}{|{\zeta_n-\zeta^*_n|}}  \Big( 
\int_{\zeta_n}^{\zeta_n^*}(\alpha_n^*-\beta_n^*)' d A_{n,1}(s) - \int_{\zeta_n}^{\zeta_n^*}
\log{r_{n}(s;\alpha^*_n,\beta_n^*)} d A_{n,0}(s)\Big)\Big| \\
&\leq
\frac{1_{{\zeta_n}<{\zeta^*_n}}}{|{\zeta_n-\zeta^*_n|}} 
\Big| \int_{\zeta_n}^{\zeta^*_n} (\alpha_n^*-\beta_n^*)'dM_{n,1}(s)\Big|
+ \frac{1_{{\zeta_n}<{\zeta^*_n}}}{|{\zeta_n-\zeta^*_n|}} 
\Big| \int_{\zeta_n}^{\zeta^*_n}
 \log{r_{n}(s;\alpha^*_n,\beta_n^*)} d M_{n,0}(s)\Big|\notag\\
 & ~~+ \frac{1_{{\zeta_n}<{\zeta^*_n}}}{|{\zeta_n-\zeta^*_n|}} 
\sup_{s\in [\zeta_n,\zeta_n^*]}\big|\log{R_{n}(s;\alpha^*_n,\beta_n^*)} 
- \log{r_{n}(s;\alpha^*_n,\beta_n^*)} \big|
  \Big| \int_{\zeta_n}^{\zeta^*_n}d M_{n,0}(s)\Big|\notag\\
& ~~+ 
 \frac{1_{{\zeta_n}<{\zeta^*_n}}}{\zeta^*_n-\zeta_n}
\sup_{s\in [\zeta_n,\zeta_n^*]}\big|\log{R_{n}(s;\alpha^*_n,\beta_n^*)} 
- \log{r_{n}(s;\alpha^*_n,\beta_n^*)} \big|
 \int_{\zeta_n}^{\zeta^*_n}d A_{n,0}(s)\notag\\
&=: a_{n,1}+a_{n,2} +a_{n,3} + a_{n,4}.\notag
\end{align}
For $a_{n,k},$ $k=1,2,3,$ by Lenglart's inequality and condition A2, we have that for any positive $\epsilon$ and $\epsilon_{n,j}, j\in {\mathbb Z}^{+}$,  there exists a constant $B>0$ such that 
\begin{eqnarray*}
\mathbb{Q}_n\Big(\sup_{m_n|\zeta_n^*-\zeta_n|>h_n} a_{n,k}>\epsilon \Big) &\leq&  
\sum_{j=1}^\infty \mathbb{Q}_n\Big(\sup_{2^{j-1}h_n\leq m_n|\zeta_n^*-\zeta_n|<2^jh_n} 
a^2_{n,k}>\epsilon^2 \Big)\\
&\leq&\sum_{j=1}^\infty \frac{\epsilon_{n,j}}{\epsilon^2} + \frac{B\rho_2}{\epsilon_{n,j} 2^j h_n},
\end{eqnarray*}
where $\rho_2$ is defined as in condition A2.
Since $\epsilon_{n,j}$ are arbitrary,  it follows that
  $\sup_{m_n|\zeta_n^*-\zeta_n|>h_n}$ $a_{n,k}\xrightarrow[]{\mathbf{P}} 0,$ $k=1,2,3$.
In addition, for $a_{n,4}$, we have 
$$\sup_{m_n|\zeta_n^*-\zeta_n|>h_n} a_{n,4} =o_{\mathbf{P}}(1)
\sup_{m_n|\zeta_n^*-\zeta_n|>h_n} 
 \frac{1_{{\zeta_n}<{\zeta^*_n}}}{\zeta^*_n-\zeta_n} \int_{\zeta_n}^{\zeta^*_n}d A_{n,0}(s)
 \xrightarrow[]{\mathbf{P}} 0.$$
Thus $\eqref{approxan}  \xrightarrow[]{\mathbf{P}} 0$. From condition A2 and Theorem \ref{theorem1}, we have
\begin{align*}
&\sup_{m_n|\zeta_n^*-\zeta_n|>h_n}
\frac{1_{{\zeta_n}<{\zeta^*_n}}}{|{\zeta_n-\zeta^*_n|}}  \biggr( 
\int_{\zeta_n}^{\zeta_n^*}(\alpha_n^*-\beta_n^*)' d A_{n,1}(s) - \int_{\zeta_n}^{\zeta_n^*}
\log{r_{n}(s;\alpha^*_n,\beta_n^*)} d A_{n,0}(s)\biggr) \\
=&\sup_{m_n|\zeta_n^*-\zeta_n|>h_n}
\frac{1_{{\zeta_n}<{\zeta^*_n}}}{|{\zeta_n-\zeta^*_n|}}  
\int_{\zeta_n}^{\zeta^*_n}\left( (\alpha_0-\beta_0)'{\bar z_n(s; \beta_0)}
- \log{r_{n}(s;\alpha_0,\beta_0)}\right)d A_{n,0}(s)+o_{\mathbf{P}}(1).
\end{align*}
Since $(\alpha_0-\beta_0)'{\bar z_n(s; \beta_0)} - \log{r_{n}(s;\alpha_0,\beta_0)}<0$ and it is continuous in a neighborhood of $\zeta_0$,  there exists a constant 
$\kappa_0<0$ such that 
 for any sequence $h_n\rightarrow\infty$ and $h_n/m_n\rightarrow 0$,
 \begin{eqnarray*}
0< -\kappa_0\rho_1\leq  \inf_{ {m_n}|\zeta_n-\zeta^*_n| >h_n}\{-a_n\} 
\leq  \sup_{ {m_n}|\zeta_n-\zeta^*_n| >h_n}\{-a_n\} \leq -\rho_2\kappa_0
\end{eqnarray*}
holds with probability tending to 1 as $n\rightarrow\infty$.
Similarly, we have 
 \begin{eqnarray*}
0<-2\kappa_0\rho_1\leq \inf_{ {m_n}|\zeta_n-\zeta^*_n| >h_n}\{-a_n-b_n\} 
\leq  \sup_{ {m_n}|\zeta_n-\zeta^*_n| >h_n}\{-a_n-b_n\} \leq -2\rho_2\kappa_0
\end{eqnarray*}
holds with probability tending to 1 as $n\rightarrow\infty$.

Combining the above derivations for \eqref{dif5}, we have that
\begin{align}\label{final}
-{m_n}|\zeta_n-\zeta^*_n| (a_n+ b_n) \leq \,&  
\sqrt{m_n}|\alpha_n^*-\alpha_n| O_{\mathbf{P}}(1)
+\sqrt{m_n}|\beta_n^*-\beta_n| O_{\mathbf{P}}(1)\notag\\
& -~ {O_{\mathbf{P}}(1)}m_n(|\alpha_n^*-\alpha_n|^2+|\beta_n^*-\beta_n|^2) \\
=~& O_{\mathbf{P}}(1).\notag
\end{align}
 Thus ${m_n}(\zeta_n-\zeta^*_n)=O_{\mathbf{P}}(1)$. 
As a consequence, \eqref{final} implies that  
\begin{eqnarray}
\sqrt{m_n}|\alpha_n^*-\alpha_n| O_{\mathbf{P}}(1)
+\sqrt{m_n}|\beta_n^*-\beta_n| O_{\mathbf{P}}(1)&&\notag\\
- {O_{\mathbf{P}}(1)}m_n(|\alpha_n^*-\alpha_n|^2+|\beta_n^*-\beta_n|^2) 
&=&  O_{\mathbf{P}}(1).\notag
\end{eqnarray}
This gives that 
 $|\sqrt{m_n}(\alpha^*_n-\alpha_n)|=O_{\mathbf{P}}(1)$ and $|\sqrt{m_n}(\beta^*_n-\beta_n)|=O_{\mathbf{P}}(1)$. 

\subsection{Proof of Theorem \ref{theorem3}}

Let 
\begin{align*}
U_{n,1} =~&  \frac{1}{\sqrt{m_n}}\sum_{i=1}^{m_n} \int_0^{\zeta_n} \left(Z_{n,i}(s) - 
{\bar z_n(s;\alpha_n)}\right) dN_{n,i}(s),\\
U_{n,2}(h_\zeta) =~& \sum_{i=1}^{m_n} \int_{\zeta_n+h_\zeta/m_n}^{\zeta_n}\left( (\beta_n-\alpha_n)'Z_{n,i}(s) - 
\log{r_{n}(s;\beta_n,\alpha_n)} \right)1_{h_\zeta<0}dN_{n,i}(s),\\
U_{n,3}(h_\zeta) =~& \sum_{i=1}^{m_n} \int_{\zeta_n+h_\zeta/m_n}^{\zeta_n}1_{h_\zeta<0}dN_{n,i}(s),\\
\end{align*}
\begin{align*}
U_{n,4}(h_\zeta) =~& \sum_{i=1}^{m_n} \int_{\zeta_n}^{\zeta_n+h_\zeta/m_n} \left( (\alpha_n-\beta_n)'Z_{n,i}(s) - 
\log{r_{n}(s;\alpha_n,\beta_n)}\right) 1_{h_\zeta > 0}dN_{n,i}(s),\\
U_{n,5}(h_\zeta) =~& \sum_{i=1}^{m_n} \int_{\zeta_n}^{\zeta_n+h_\zeta/m_n} 1_{h_\zeta > 0}dN_{n,i}(s),\\
U_{n,6} =~& \frac{1}{\sqrt{m_n}}\sum_{i=1}^{m_n} \int_{\zeta_n}^\tau \left(Z_{n,i}(s) - 
{\bar z_n(s;\beta_n)}\right) dN_{n,i}(s).
\end{align*}
We define processes $ J_n(h_\zeta) := U_{n,3}(h_\zeta) + U_{n,5}(h_\zeta)$ and
\begin{eqnarray*}
 U_n(h) &:=& h_\alpha' U_{n,1} - \frac{1}{2}h_\alpha'\Big(\int_0^{\zeta_0} Q(s; \alpha_0)s_0(s;\alpha_0)\lambda_0(s)ds\Big)h_\alpha + U_{n,2}(h_\zeta)\\
&&+ h_\beta' U_{n,6} - \frac{1}{2}h_\beta'\Big(\int_{\zeta_0}^\tau Q(s; \beta_0)s_0(s;\beta_0)\lambda_0(s)ds\Big)h_\beta + U_{n,4}(h_\zeta).
\end{eqnarray*}

The limit law of $ U_n$ and $J_n$ can be deduced from that  of $U_{n,i}$ and  is given as follows.

\begin{lemma}\label{lemma3}
Let $K\subset \mathbb{R}$ be a compact interval and $\Theta = \tilde \Theta\times K\subset \mathbb{R}^{2p+1}$ a compact set.
Then, under conditions $A1$-$A3$,  $( U_n,  J_n)$ converges weakly in the Skorohod topology to $(U,J)$ in ${\cal D}_\Theta\times {\cal D}_{K}$. 
\end{lemma}

Next we show that processes $U_n$ and $U^*_n$ have the same asymptotic distribution. 

\begin{lemma}\label{lemma2}
Let $\Theta$ be a compact set in $\mathbb{R}^{2p+1}$. Then under conditions $A1$ and $A2$,
$$\sup_{h\in\Theta}|U_n(h)-U^*_n(h)| \xrightarrow[]{\mathbf{P}} 0.$$
\end{lemma}

Thus, $( U^*_n,  J^*_n)$ also converges weakly in the Skorohod topology to $(U,J)$ in ${\cal D}_\Theta\times {\cal D}_{K}$. Then by Theorem 3.1 in \cite{SeijoSen2011b}, we have the desired conclusion. 

\section{Appendix}\label{Appendix}
This appendix contains proofs of Lemmas \ref{lemma4}, \ref{lemmancc}, \ref{lemma3}, \ref{lemma2} and \ref{lemma1}.

\begin{proof}[Proof of Lemma \ref{lemma4}]
From the definition of $U$ it is easily seen that 
\begin{eqnarray*}
\phi_\alpha &=& \biggr(\int_0^{\zeta_0} Q(s;\alpha_0)s_0(s;\alpha_0)\lambda_0(s)ds\biggr)^{-1}U_1, \\
\phi_\beta  &= & \left(\int_{\zeta_0}^\tau Q(s;\beta_0)s_0(s;\beta_0)\lambda_0(s)ds\right)^{-1}U_6,\\
\phi_\zeta  & =& \mbox{sargmax}_{h\in\mathbb{R}^{2p+1}}\{U_2(h_\zeta) + U_4(h_{\zeta})\}.
\end{eqnarray*}
Due to the independence of $U_1, U_2, U_4$ and $U_6$, $\phi_\alpha$, $\phi_\beta$, and $\phi_\zeta$ are independent.  In addition,  \eqref{dis1} and \eqref{dis2} hold. 

We now show the existence  of  $\phi_\zeta$. It suffices to show that $U_2(h_\zeta) + U_4(h_{\zeta}) \rightarrow -\infty$ as $|h_\zeta|\rightarrow \infty$. For $h_\zeta>0$,
\begin{align*}
U_4(h_\zeta) 
=&-\Gamma^+(h_\zeta)\log{r(\zeta_0;\alpha_0,\beta_0)} 
+\sum_{1\leq i\leq\Gamma^+(h_\zeta)} (\alpha_0-\beta_0)v_i^+ \\
=& \sum_{1\leq i\leq\Gamma^+(h_\zeta)} \left\{
 (\alpha_0-\beta_0)v_i^+ 
- \mathbf{E}\left[ (\alpha_0-\beta_0)v_i^+ 
\right]\right\}\\
&~~~~+
\mathbf{E}\left[  (\alpha_0-\beta_0)v_i^+ -\log{r(\zeta_0;\alpha_0,\beta_0)}\right]
 \Gamma^+(h_\zeta).
\end{align*}
Since $\Gamma^+(h_\zeta)\xrightarrow[]{}\infty$ as $h_\zeta\rightarrow\infty$ and 
\begin{eqnarray*}\label{inf}
   \mathbf{E}\left[ (\alpha_0-\beta_0)v_i^+ -\log{r(\zeta_0;\alpha_0,\beta_0)} \right]
=
 {(\alpha_0-\beta_0)'z({\zeta_0;\beta_0})}
  -\log{r(\zeta_0;\alpha_0,\beta_0)} <0,
\end{eqnarray*}
 $U_{4}(h_\zeta)\xrightarrow[]{}-\infty$ as $h_\zeta\rightarrow\infty$.
A similar argument gives $U_2(h_\zeta)\rightarrow -\infty$ as $h_\zeta\rightarrow -\infty$, which completes the proof. 
\end{proof}

\begin{proof}[Proof of Lemma \ref{lemmancc}]
For \eqref{1e}, we only need to show that the first sequence does not converge in probability. 
Take $\epsilon <1/4$. From Theorem \ref{theorem2}, there exists a constant $B_\epsilon>0$ such that  $P(n|\hat\zeta_n-\zeta_0|\leq B_\epsilon)>1-\epsilon$ for all large $n$. Choose $h>2B_\epsilon$ and let 
\begin{eqnarray*}
&&\hat E_n = \sum_{i=1}^n \int_{\hat\zeta_n}^{\hat\zeta_n +\frac{h}{n}} 
dN_{i}\left(s\right),~~
E_{n,1} = \sum_{i=1}^n \int_{\zeta_0+\frac{B_\epsilon}{n} }^{\zeta_0 +\frac{h-B_\epsilon}{n}} 
dN_{i}\left(s\right),\,\\
&&E_{n,2} =  \sum_{i=1}^n \int_{\zeta_0-\frac{B_\epsilon}{n}}^{\zeta_0 +\frac{h+B_\epsilon}{n}} 
dN_{i}\left(s\right).
\end{eqnarray*} 
Then, \begin{equation}\label{E}
P(E_{n,1}\leq \hat E_n\leq E_{n,2})
\geq P(n|\hat\zeta_n-\zeta_0|\leq B_\epsilon)>1-\epsilon.
\end{equation}
We know that for any $h_1<h_2$,
\begin{eqnarray*}
\sum_{i=1}^n  \int_{\zeta_0+\frac{h_1}{n}}^{\zeta_0 +\frac{h_2}{n}} dN_i(s)
\leadsto \mbox{Poisson}\big(\lambda_0(\zeta_0)[(h_2-h_1\vee 0)s_0(\zeta_0;\beta_0)
-(h_1\wedge 0)s_0(\zeta_0;\alpha_0)]\big).~
\end{eqnarray*}
Therefore, 
$E_{n,1} \leadsto \mbox{Poisson}(\lambda_0(\zeta_0)(h-2B_\epsilon)s_0(\zeta_0;\beta_0))$
and 
$E_{n,2} \leadsto \mbox{Poisson}(\lambda_0(\zeta_0)(h+B_\epsilon)s_0(\zeta_0;\beta_0)+\lambda_0(\zeta_0)B_\epsilon s_0(\zeta_0;\alpha_0)).$ Then by Lemma A.4 in \cite{SeijoSen2011},
 there is a constant $h_0$ such that when $h>h_0$, we can find two numbers $N_{1,h}<N_{2,h}\in \mathbb{N}$ satisfying 
$$\liminf_{n\rightarrow\infty} \mathbf{P}(E_{n,1} > N_{2,h}) >2\epsilon ~\mbox{and}~ 
\liminf_{n\rightarrow\infty} \mathbf{P}(E_{n,2} < N_{1,h}) >2\epsilon.$$
Combining with \eqref{E}, we have 
$$\mathbf{P}(\hat E_n\geq E_{n,1} > N_{2,h}, ~i.o.) >\epsilon ~\mbox{and}~ 
 \mathbf{P}(\hat E_n \leq E_{n,2} < N_{1,h}, ~i.o.) >\epsilon.$$
Then by the Hewitt-Savage 0-1 law, the permutation invariant events 
$\{\hat E_n> N_{2,h}, ~i.o.\}$ and $\{\hat E_n< N_{1,h}, ~i.o.\}$ occur with probability 1, which implies that $\hat E_n$ does not have an almost sure limit. 
A similar argument applies for any increasing sequence of natural numbers $\{n_k\}_{k=1}^\infty$ and gives that 
  $\hat E_n$ does not converge in probability. 
  
 For \eqref{2e},  consider  the real part of $\phi_i$, $Re(\phi_i)$, and define
$\hat E^{\phi}_n = \sum_{i=1}^n \int_{\hat\zeta_n}^{\hat\zeta_n +\frac{h}{n}} Re(\phi_i(s))$ 
$dN_{i}\left(s\right)$,
$E^{\phi}_{n,1} = \sum_{i=1}^n \int_{\zeta_0+\frac{B_\epsilon}{n} }^{\zeta_0 +\frac{h-B_\epsilon}{n}} 
Re(\phi_i(s))dN_{i}\left(s\right)$, and
$E^{\phi}_{n,2} =  \sum_{i=1}^n \int_{\zeta_0-\frac{B_\epsilon}{n}}^{\zeta_0 +\frac{h+B_\epsilon}{n}} 
Re(\phi_i(s))$ $dN_{i}\left(s\right).$ 
It is sufficient to show that $\hat E^{\phi}_n $ does not converge. 
Since for any $h_1<h_2$,
$\sum_{i=1}^n  \int_{\zeta_0+\frac{h_1}{n}}^{\zeta_0 +\frac{h_2}{n}}$ $Re(\phi_i(s))dN_i(s)$ converges to a compound Poisson distribution, then a similar argument as above gives the desired conclusion. 
\end{proof}

\begin{proof}[Proof of Lemma \ref{lemma3}]


It is sufficient to show the weak convergence in probability of  $(U_{n,1},\cdots, U_{n,6})$ to  $(U_{1},\cdots, U_{6})$.
We first prove the convergence of  its finite dimensional joint characteristic function.  Consider  real numbers $h_{-N}<\cdots<h_{-1} <0 = h_0 <h_1<\cdots<h_N$ and the linear combination
\begin{align*}
W_n =& ~
\mu U_{n,1} + v U_{n,6}\\
&+\sum_{-N\leq j\leq -1} \left\{q_j (U_{n,2}(h_j) - U_{n,2}(h_{j+1})) +p_j (U_{n,3}(h_j) - U_{n,3}(h_{j+1})) \right\}\\
& + \sum_{1\leq j\leq N} \left\{q_j (U_{n,4}(h_j) - U_{n,4}(h_{j-1})) +p_j (U_{n,5}(h_j) - U_{n,5}(h_{j+1}))\right\} ,
\end{align*}
where $p_j,q_j \in\mathbb{R},  j=-N,\cdots,N,$ and $\mu, v\in \mathbb{R}^{1\times p}$. 
For simplicity, we write $h_{j,n}=h_j/m_n$.

The characteristic function of $W_n$ is ${\mathbf E} [e ^{\imath tW_n}] $ and can be expressed as
\begin{align*}
& {\mathbf E}\biggr[ \exp\biggr(\imath t\mu\frac{1}{\sqrt{m_n}}\sum_{k=1}^{m_n} \int_0^{\zeta_n} \left(Z_{n,k}(s) - 
{\bar z_n(s;\alpha_n)}\right) dN_{n,k}(s) \\
&+\imath tv \int_{\zeta_n}^\tau \left(Z_{n,k}(s) - 
{\bar z_n(s;\beta_n)}\right) dN_{n,k}(s)\\
& +\imath t \sum_{j=-N}^{-1} \sum_{k=1}^{m_n} 
\int_{\zeta_n+h_{j,n}}^{\zeta_n+h_{j+1,n}}\left(q_j (\beta_n-\alpha_n)'Z_{n,k}(s) - 
q_j\log{r_{n}(s;\beta_n,\alpha_n)} + p_j\right)dN_{n,k}(s) \\
& +\imath t \sum_{j=1}^{N}  \sum_{k=1}^{m_n} \int_{\zeta_n+h_{j-1,n}}^{\zeta_n+h_{j,n}} 
\left(q_j (\alpha_n-\beta_n)'Z_{n,k}(s) - q_j
\log{r_{n}(s;\alpha_n,\beta_n)} +p_j\right) dN_{n,k}(s)\biggr)\biggr].
\end{align*}
By the independence of the observations $\{(\tilde T_{n,k}, \delta_{n,k},Z_{n,k}): k=1,\cdots,m_n\}$, ${\mathbf E} [e ^{\imath tW_n}] $ can be further written as 
\begin{align*}
&\prod_{k=1}^{m_n} {\mathbb Q}_n\biggr\{ 
1+\int_0^{\tau}\Big[e^{\imath t\frac{1}{\sqrt{m_n}} \mu \big( Z_{n,k}(s) - 
{\bar z_n(s;\alpha_n)}\big)1_{s<\zeta_n}}
 -1 \Big] dN_{n,k}(s)\\
&~~~~~~+ \int_0^{\tau} \Big[
e^{\imath tv\frac{1}{\sqrt{m_n}}  \big( Z_{n,k}(s) - {\bar z_n(s;\beta_n)}\big)
 1_{s>\zeta_n}}-1 \Big] dN_{n,k}(s)\\
&~~~~~~ +\sum_{-N\leq j\leq -1}\int_{\zeta_n+h_{j,n}}^{\zeta_n+h_{j+1,n}}\Big[e^{\imath t \left(q_j (\beta_n-\alpha_n)'Z_{n,k}(s) - 
q_j\log{r_{n}(s;\beta_n,\alpha_n)} + p_j\right)}-1\Big]dN_{n,k}(s) \\
& ~~~~~~ + \sum_{1\leq j\leq N}  \int_{\zeta_n+h_{j-1,n}}^{\zeta_n+h_{j,n}} 
\Big[e^{\imath t\left(q_j (\alpha_n-\beta_n)'Z_{n,k}(s) - q_j
\log{r_{n}(s;\alpha_n,\beta_n)} +p_j\right)} -1 \Big]dN_{n,k}(s)\biggr\}.
\end{align*}
For the first two integrals, take Taylor's expansions of the exponential functions and we have that 
${\mathbf E} [e ^{\imath t W_n}]$ equals 
\begin{align*}
& 
\prod_{k=1}^{m_n}\biggr\{ 1+ {\mathbb Q}_n \left(
\int_0^{\tau}\frac{1}{\sqrt{m_n}}it\mu \Big( Z_{n,k}(s) - 
{\bar z_n(s;\alpha_n)}\Big)1_{s<\zeta_n}
dN_{n,k}(s)\right)\\
&~+ {\mathbb Q}_n \left(
\int_0^{\tau} \frac{1}{\sqrt{m_n}} itv \Big( Z_{n,k}(s) 
- {\bar z_n(s;\beta_n)}\Big)1_{s>\zeta_n}  dN_{n,k}(s)\right)
\\
&~ - \frac{1}{{2m_n}} {\mathbb Q}_n 
\biggr(\int_0^{\tau}t^2
\mu \Big( Z_{n,k}(s) - {\bar z_n(s;\alpha_n)}\Big)^{\otimes 2}\mu'
~1_{s<\zeta_n}dN_{n,k}(s)\biggr)\\
&~ - \frac{1}{{2m_n}} {\mathbb Q}_n\biggr(\int_0^{\tau}t^2v
\Big( Z_{n,k}(s) - {\bar z_n(s;\beta_n)}\Big)^{\otimes 2}v'
~1_{s>\zeta_n} dN_{n,k}(s)\biggr) +o(m_n^{-1})\\
& + \sum_{-N\leq j\leq -1}{\mathbb Q}_n\biggr(
\int_{\zeta_n+h_{j,n}}^{\zeta_n+h_{j+1,n}}
\Big[e^{\imath t \left(q_j (\beta_n-\alpha_n)'Z_{n,k}(s) - 
q_j\log{r_{n}(s;\beta_n,\alpha_n)} + p_j\right)}-1\Big]dN_{n,k}(s)\biggr) \\
&  + 
\sum_{1\leq j\leq N}  {\mathbb Q}_n\biggr(
\int_{\zeta_n+h_{j-1,n}}^{\zeta_n+h_{j,n}} 
\Big[e^{\imath t \left(q_j (\alpha_n-\beta_n)'Z_{n,k}(s) - q_j
\log{r_{n}(s;\alpha_n,\beta_n)} +p_j\right)} -1 \Big]dN_{n,k}(s)\biggr)\biggr\}.
\end{align*}
By condition A2,  
$\mathbb{Q}_n\big(\sum_{i=1}^{m_n}\int_0^{\zeta_n} \big( Z_{n,k}(s) - {\bar z_n(s;\alpha_n)}\big)dN_{n,k}(s)\big)$ and
$\mathbb{Q}_n\big(\sum_{i=1}^{m_n}\int_{\zeta_n}^\tau \big( Z_{n,k}(s) - {\bar z_n(s;\beta_n)}\big)dN_{n,k}(s)\big)$ converge to 0. Thus
 ${\mathbf E}[e ^{\imath tW_n}]$ equals
\begin{align*}
&(1+o(1))\exp\biggr\{-
 \frac{1}{{2}} \sum_{k=1}^{m_n} \mathbb{Q}_n
 \biggr(\int_0^{\zeta_n} t^2\mu
 \big( Z_{n,k}(s) - {\bar z_n(s;\alpha_n)}\big)^{\otimes 2}\mu'~
dN_{n,k}(s)\biggr)\\
&~~~~~~~~~~~~
- \frac{1}{{2}} \sum_{k=1}^{m_n} \mathbb{Q}_n
 \biggr(\int_{\zeta_n}^\tau
  t^2 v \big( Z_{n,k}(s) - {\bar z_n(s;\beta_n)}\big)^{\otimes 2}v
~dN_{n,k}(s) \biggr)\biggr\}\\
\times&\exp\biggr\{\sum_{j=-N}^{-1}
\mathbb{Q}_n \biggr(\sum_{k=1}^{m_n} 
\int_{\zeta_n+h_{j,n}}^{\zeta_n+h_{j+1,n}}\Big[e^{\imath t \left(q_j (\beta_n-\alpha_n)'Z_{n,k}(s) - 
q_j\log{r_{n}(s;\beta_n,\alpha_n)} + p_j\right)}-1\Big] dN_{n,k}(s)
\biggr)\biggr\}\\
\times& \exp\biggr\{ \sum_{j=1}^{N} 
\mathbb{Q}_n \biggr(\sum_{k=1}^{m_n} 
\int_{\zeta_n+h_{j-1,n}}^{\zeta_n+h_{j,n}} 
\Big[e^{\imath t\left(q_j (\alpha_n-\beta_n)'Z_{n,k}(s) - q_j
\log{r_{n}(s;\alpha_n,\beta_n)} +p_j\right)} -1 \Big] dN_{n,k}(s)
\biggr)\biggr\}.
\end{align*}
It is easily seen that the first exponential component  in the above display converges to ${\mathbf E} [e ^{\imath t\mu U_{1} +\imath tvU_{6}}]$. Therefore, Lemma \ref{lemma1}
together with condition $A3$ implies that 
\begin{align*}
{\mathbf E} [e ^{\imath tW_{n}}]
=~&(1+o(1)) {\mathbf E} [e ^{\imath t\mu U_{1} +\imath tvU_{6}}]\\
&\times \exp\biggr\{  \lambda_0(\zeta_0)
\sum_{-N\leq j\leq -1} (h_{j+1}-h_{j})\\
&~~~~~~\times\Big[e^{-q_j\log{r(\zeta_0;\beta_0,\alpha_0)} + p_j} 
 ~s_{0}(\zeta_0; itq_j (\beta_0-\alpha_0)+\alpha_0)
- s_{0}(\zeta_0; \alpha_0)\Big]  \biggr\}\\
&\times \exp\biggr\{ \lambda_0(\zeta_0)
\sum_{1\leq j\leq N}  (h_{j}-h_{j-1})\\
&~~~~~~\times \Big[ e^{- q_j\log{r(\zeta_0;\alpha_0,\beta_0)} +p_j}
  ~s_{0}(\zeta_0; itq_j (\alpha_0-\beta_0)+\beta_0)
-s_{0}(\zeta_0; \beta_0)\Big]  \biggr\}.
\end{align*}

For $(U_1,\cdots, U_6)$ as defined in Section 3.2, define the linear combination
\begin{align*}
W :=& 
\mu U_{1} + v U_{6}+
\sum_{-N\leq j\leq -1} \left\{q_j (U_{n,2}(h_j) - U_{2}(h_{j+1})) +p_j (U_{3}(h_j) - U_{n,3}(h_{j+1})) \right\}\\
& + \sum_{1\leq j\leq N} \left\{q_j (U_{n,4}(h_j) - U_{4}(h_{j-1})) +p_j (U_{5}(h_j) - U_{n,5}(h_{j+1}))\right\} .
\end{align*}
By the definition of $(U_1,\cdots, U_6)$, we know that the characteristic function of $W$ has the same form as the limit of ${\mathbf E}[e ^{\imath tW_{n}}]$. Thus, we have the  weak convergence of the finite dimensional distributions. 

To further prove the weak convergence of  $(U_1,\cdots, U_6)$, we use Theorem 15.6 in \cite{Billingsley1968}.
It's sufficient to show for each $U_{n,i}, i=2,3,4,5,$ there exists a nondecreasing, continuous function $F$ such that for any $h_1<h<h_2$,
\begin{equation}\label{bll}
{\mathbf E} |U_{n,i}(h_1)-U_{n,i}(h)| |U_{n,i}(h_2)-U_{n,i}(h)|  \leq  (F(h_2)-F(h_1))^2.
\end{equation}
Consider $U_{n.2}$. For  $h_1<h<h_2<0$,
\begin{eqnarray*}
&& {\mathbf E} |U_{n,2}(h_1)-U_{n,2}(h)| |U_{n,2}(h_2)-U_{n,2}(h)| \\
&\leq& 
{\mathbf E} \sum_{i=1}^{m_n} \int_{\zeta_n+h_1/m_n}^{\zeta_n+h/m_n}\left| (\beta_n-\alpha_n)'Z_{n,i}(s) - 
\log{r_{n}(s;\beta_n,\alpha_n)} \right| dN_{n,i}(s)\\
&&~~~~\times\sum_{i=1}^{m_n} \int_{\zeta_n+h/m_n}^{\zeta_n+h_2/m_n}\left| (\beta_n-\alpha_n)'Z_{n,i}(s) - 
\log{r_{n}(s;\beta_n,\alpha_n)} \right| dN_{n,i}(s)\\
&\leq& \sup_{1\leq i\leq m_n\atop s\in[\zeta_n+\frac{h_1}{{m_n}},\zeta_n+\frac{h_2}{{m_n}}]}
 \left| (\beta_n-\alpha_n)'Z_{n,i} (s) - \log{r_{n}(s;\beta_n,\alpha_n)} \right| \\
&&~~~~~~~~~~~\times m_n^2
 {\mathbf E} \biggr[\int_{\zeta_n+h_1/m_n}^{\zeta_n+h/m_n}dN_{n,i}(s)\biggr]
  {\mathbf E}\biggr[  \int_{\zeta_n+h/m_n}^{\zeta_n+h_2/m_n} dN_{n,i}(s)\biggr]\\
 &\leq & B |h_2-h_1|^2,
\end{eqnarray*}
where $B>0$ is some constant.
Thus, \eqref{bll} holds for $U_{n,2}$. Similar arguments give that  \eqref{bll} is satisfied for  $U_{n,i}, i=3,4,5.$ Then our conclusion follows from Theorem 15.6 in \cite{Billingsley1968}.
\end{proof}

\begin{proof}[Proof of Lemma \ref{lemma2}]
For notational simplicity, we write 
$$h_{\alpha,n}=\frac{h_\alpha}{\sqrt{m_n}},
h_{\beta,n}=\frac{h_\beta}{\sqrt{m_n}}, h_{\zeta,n}=\frac{h_\zeta}{{m_n}}.$$
We start by writing $U_n^*$ as follows:
\begin{eqnarray*}
 U^*_n(h) 
 := u_{n,1} (h)+u_{n,2}(h) +u_{n,3}(h)+u_{n,4}(h)\end{eqnarray*}
where 
\begin{align*}
 u_{n,1} (h) =& \sum_{i=1}^{m_n} \int_0^{\tau}\big({h_{\alpha,n}'}Z_{n,i}(s) - 
\log{R_{n}(s;\alpha_n+h_{\alpha,n},\alpha_n)}\big)
1_{s\leq \zeta_n\wedge(\zeta_n+h_{\zeta,n}) }dN_{n,i}(s),\\
 u_{n,2} (h)=& \sum_{i=1}^{m_n} \int_0^\tau \big( (\alpha_n-\beta_n+h_{\alpha,n})' Z_{n,i}(s) \\
&~~~~~~~~~~-\log{R_{n}(s;\alpha_n+h_{\alpha,n},\beta_n)}\big)
1_{\zeta_n<s\leq \zeta_n+h_{\zeta,n} }dN_{n,i}(s),\\
 u_{n,3} (h)=& \sum_{i=1}^{m_n} \int_0^\tau \big( (\beta_n-\alpha_n+h_{\beta,n})' Z_{n,i}(s) \\
 &~~~~~~~~~~- 
\log{R_{n}(s;\beta_n+h_{\beta,n},\alpha_n)}\big)
1_{\zeta_n+h_{\zeta,n}<s\leq {\zeta_n}} dN_{n,i}(s),\\
u_{n,4}(h)=&\sum_{i=1}^{m_n} \int_0^{\tau}\big(h_{\beta,n}'Z_{n,i}(s)- 
\log{R_{n}(s;\beta_n+h_{\beta,n},\beta_n)}\big)
1_{s> \zeta_n\vee(\zeta_n+h_{\zeta,n}) }dN_{n,i}(s).
\end{align*}
For $h=(h_\alpha', h_\beta', h_\zeta)' \in\Theta$, consider the difference between $u_{n,1}$ and the first two terms in $U_n$:
\begin{eqnarray}\label{Uncal}
&&\Big| u_{n,1} (h)-h_\alpha' U_{n,1} + \frac{1}{2}h_\alpha'\Big(\int_0^{\zeta_0} Q(s; \alpha_0)s_0(s;\alpha_0)\lambda_0(s)ds\Big)h_\alpha  \Big|\notag\\
&\leq& \biggr|\sum_{i=1}^{m_n} \int_0^{\tau}\big({h_{\alpha,n}'}Z_{n,i}(s) - 
\log{R_{n}(s;\alpha_n+h_{\alpha,n},\alpha_n)}\big)
1_{ \zeta_n\wedge(\zeta_n+{h_{\zeta,n}}) <s\leq \zeta_n}dN_{n,i}(s)\biggr|\notag\\
&& + \biggr|\sum_{i=1}^{m_n} \int_0^{\zeta_n}\big(h_{\alpha,n}'Z_{n,i}(s) - 
\log{R_{n}(s;\alpha_n+h_{\alpha,n},\alpha_n)}\big)dN_{n,i}(s)-h_\alpha' U_{n,1}\notag \\
&&~~~~ + \frac{1}{2}h_\alpha'\biggr(\int_0^{\zeta_0} Q(s; \alpha_0)s_0(s;\alpha_0)\lambda_0(s)ds\biggr)h_\alpha\biggr|.
\end{eqnarray}
It is easily seen that 
\begin{align}\label{cross}
& \sup_{h\in \Theta} \biggr|\sum_{i=1}^{m_n} \int_0^{\tau}\left(h_{\alpha,n}' Z_{n,i}(s) - 
\log{R_{n}(s;\alpha_n+h_{\alpha,n},\alpha_n)}\right)
1_{ \zeta_n\wedge(\zeta_n+h_{\zeta,n}) <s\leq \zeta_n}dN_{n,i}(s)\biggr|\notag\\
\leq~& \frac{B_{1}}{\sqrt{m_n}} \sum_{i=1}^{m_n} \int_0^\tau1_{ \zeta_n- \frac{B_2}{{m_n}} <s\leq \zeta_n}dN_{n,i}(s) \xrightarrow[]{\mathbf{P}} 0,
\end{align}
where $B_1$ and $B_2$ are some constants. On the other hand, by Taylor's expansion, we have that 
\begin{eqnarray*}
&& \sum_{i=1}^{m_n} \int_0^{\zeta_n}\left(h_{\alpha,n}'Z_{n,i}(s) - 
\log{R_{n}(s;\alpha_n+h_{\alpha,n},\alpha_n)}\right)dN_{n,i}(s)\\
&=& h_\alpha'U_{n,1} - \frac{1}{2}h_\alpha'\Big(\frac{1}{n}\sum_{i=1}^{m_n} \int_0^{\zeta_n} Q_n(s; \alpha_n)dN_i(s)\Big)h_\alpha +o(1).
\end{eqnarray*}
Then by condition A1 and the uniform convergence of $Q_n$, the second term of \eqref{Uncal}  converges to 0 uniformly in probability. Thus, 
 \begin{eqnarray*}
\sup_{h\in\Theta}\Big| u_{n,1} (h)-h_\alpha' U_{n,1} + \frac{1}{2}h_\alpha'\Big(\int_0^{\zeta_0} Q(s; \alpha_0)s_0(s;\alpha_0)\lambda_0(s)ds\Big)h_\alpha  \Big| \xrightarrow[]{\mathbf{P}} 0.
\end{eqnarray*}
A similar argument gives that
$|u_{n,2} - U_{n,4}|\xrightarrow[]{\mathbf{P}} 0, ~  |u_{n,3} - U_{n,2}|\xrightarrow[]{\mathbf{P}} 0,$
and 
$$\sup_{h\in\Theta}\Big| u_{n,4} (h)- h_\beta' U_{n,6} - \frac{1}{2}h_\beta'\Big(\int_{\zeta_0}^\tau Q(s; \beta_0)s_0(s;\beta_0)\lambda_0(s)ds\Big)h_\beta \Big|\xrightarrow[]{\mathbf{P}} 0. $$
This completes our proof. 
\end{proof}

\begin{lemma}\label{lemma1}
Under condition $A1$, for $k=1,2,$ and $3$, as $n\rightarrow\infty$,
$$\sup_{t\in [0,\tau]\atop \gamma\in  \Theta_\alpha\cup\Theta_\beta}|S_{n,k}(t;\gamma) - s_{n,k}(t;\gamma)| \rightarrow 0
\mbox{ and} \sup_{t\in [0,\tau]\atop \gamma\in  \Theta_\alpha\cup\Theta_\beta}|S_{n,k}(t;\gamma) - s_{k}(t;\gamma)| \rightarrow 0.$$
\end{lemma}

\begin{proof}[Proof of Lemma \ref{lemma1}]
We only need to show the first convergence result. The second result then follows from condition A1. 
We start with  the case  $k=0$.
In view of Theorem 8.3 of \cite{Pollard1990}, it suffices to show that 
$\{Y_{n,i}(t)e^{\gamma'Z_{n,i}(t)}\}$ is manageable. 
Since the total variation of $Z_{n,i}$ is bounded, we can write  $e^{\gamma'Z_{n,i}(t)} = e^{\gamma' Z_{n,i}^+(t)-\gamma' Z_{n,i}^-(t)}$, where componentwise, $ Z_{n,i}^+(t)$ and 
$Z_{n,i}^-(t)$ are nonnegative, nonincreasing, and bounded by some constant $B$.  By Lemma A.2 in \cite*{BGY}, we have the manageability of $\{Y_{n,i}(t)\}$, $ \{Z_{n,i}^+(t)\}$ and 
$\{Z_{n,i}^-(t)\}$.  Then, (5.2) in \cite{Pollard1990} implies  $\{\gamma'Z_{n,i}(t)\}$ is manageable, and further $\{Y_{n,i}(t)e^{\gamma'Z_{n,i}(t)}\}$ is manageable. Thus,
$$\sup_{t\in [0,\tau], \gamma\in  \Theta_\alpha\cup\Theta_\beta}|S_{n,0}(t;\gamma) - s_{n,0}(t;\gamma)| \rightarrow 0.$$
A similar argument yields the result for $k=1$ and $2$.
\end{proof}

\bibliographystyle{rss}
\bibliography{sequential,mybib}

\end{document}